\documentclass{article}

\usepackage[final]{arxiv}

\usepackage{amsmath, amssymb,amsthm}
\usepackage[ruled,vlined,linesnumbered]{algorithm2e}
\usepackage{multirow}
\newtheorem{theorem}{Theorem}
\newtheorem{lemma}{Lemma}
\newtheorem{remark}{Remark}
\newtheorem{proposition}{Proposition}

\usepackage{graphicx}
\usepackage{caption}
\usepackage{subcaption}
\usepackage{tikz}
\usepackage{pgfplots}
\usetikzlibrary{spy,calc}

\newif\ifblackandwhitecycle
\gdef\patternnumber{0}

\pgfkeys{/tikz/.cd,
    zoombox paths/.style={
        draw=orange,
        line width=0.5pt
    },
    black and white/.is choice,
    black and white/.default=static,
    black and white/static/.style={ 
        draw=white,   
        zoombox paths/.append style={
            draw=white,
            postaction={
                draw=black,
                loosely dashed
            }
        }
    },
    black and white/static/.code={
        \gdef\patternnumber{1}
    },
    black and white/cycle/.code={
        \blackandwhitecycletrue
        \gdef\patternnumber{1}
    },
    black and white pattern/.is choice,
    black and white pattern/0/.style={},
    black and white pattern/1/.style={    
            draw=white,
            postaction={
                draw=black,
                dash pattern=on 2pt off 2pt
            }
    },
    black and white pattern/2/.style={    
            draw=white,
            postaction={
                draw=black,
                dash pattern=on 4pt off 4pt
            }
    },
    black and white pattern/3/.style={    
            draw=white,
            postaction={
                draw=black,
                dash pattern=on 4pt off 4pt on 1pt off 4pt
            }
    },
    black and white pattern/4/.style={    
            draw=white,
            postaction={
                draw=black,
                dash pattern=on 4pt off 2pt on 2 pt off 2pt on 2 pt off 2pt
            }
    },
    zoomboxarray inner gap/.initial=5pt,
    zoomboxarray columns/.initial=2,
    zoomboxarray rows/.initial=2,
    subfigurename/.initial={},
    figurename/.initial={zoombox},
    zoomboxarray/.style={
        execute at begin picture={
            \begin{scope}[
                spy using outlines={%
                    zoombox paths,
                    width=\imagewidth*1.05 / \pgfkeysvalueof{/tikz/zoomboxarray columns} - (\pgfkeysvalueof{/tikz/zoomboxarray columns} - 1) / \pgfkeysvalueof{/tikz/zoomboxarray columns} * \pgfkeysvalueof{/tikz/zoomboxarray inner gap} -\pgflinewidth,
                    height=\imageheight / \pgfkeysvalueof{/tikz/zoomboxarray rows} - (\pgfkeysvalueof{/tikz/zoomboxarray rows} - 1) / \pgfkeysvalueof{/tikz/zoomboxarray rows} * \pgfkeysvalueof{/tikz/zoomboxarray inner gap}-\pgflinewidth,
                    magnification=3,
                    every spy on node/.style={
                        zoombox paths
                    },
                    every spy in node/.style={
                        zoombox paths
                    }
                }
            ]
        },
        execute at end picture={
    \end{scope}
    \gdef\patternnumber{0}
},
        spymargin/.initial=0.5em,
        zoomboxes xshift/.initial=1,
        zoomboxes right/.code=\pgfkeys{/tikz/zoomboxes xshift=1},
        zoomboxes left/.code=\pgfkeys{/tikz/zoomboxes xshift=-1},
        zoomboxes yshift/.initial=0,
        zoomboxes above/.code={
            \pgfkeys{/tikz/zoomboxes yshift=1},
            \pgfkeys{/tikz/zoomboxes xshift=0}
        },
        zoomboxes below/.code={
            \pgfkeys{/tikz/zoomboxes yshift=-1},
            \pgfkeys{/tikz/zoomboxes xshift=0}
        },
        caption margin/.initial=4ex,
    },
    adjust caption spacing/.code={},
    image container/.style={
        inner sep=0pt,
        at=(image.north),
        anchor=north,
        adjust caption spacing
    },
    zoomboxes container/.style={
        inner sep=0pt,
        at=(image.north),
        anchor=north,
        name=zoomboxes container,
        xshift=\pgfkeysvalueof{/tikz/zoomboxes xshift}*(\imagewidth+\pgfkeysvalueof{/tikz/spymargin}),
        yshift=\pgfkeysvalueof{/tikz/zoomboxes yshift}*(\imageheight+\pgfkeysvalueof{/tikz/spymargin}+\pgfkeysvalueof{/tikz/caption margin}),
        adjust caption spacing
    },
    calculate dimensions/.code={
        \pgfpointdiff{\pgfpointanchor{image}{south west} }{\pgfpointanchor{image}{north east} }
        \pgfgetlastxy{\imagewidth}{\imageheight}
        \global\let\imagewidth=\imagewidth
        \global\let\imageheight=\imageheight
        \gdef\columncount{2}
        \gdef\rowcount{2}
        
    },
    image node/.style={
        inner sep=0pt,
        name=image,
        anchor=south west,
        append after command={
            [calculate dimensions]
            node [image container,subfigurename=\pgfkeysvalueof{/tikz/figurename}-image] {\phantomimage}
            node [zoomboxes container,subfigurename=\pgfkeysvalueof{/tikz/figurename}-zoom] {\phantomimage}
        }
    },
    color code/.style={
        zoombox paths/.append style={draw=#1}
    },
    connect zoomboxes/.style={
    spy connection path={\draw[draw=none,zoombox paths] (tikzspyonnode) -- (tikzspyinnode);}
    },
    help grid code/.code={
        \begin{scope}[
                x={(image.south east)},
                y={(image.north west)},
                font=\footnotesize,
                help lines,
                overlay
            ]
            \foreach \x in {0,1,...,9} { 
                \draw(\x/10,0) -- (\x/10,1);
                \node [anchor=north] at (\x/10,0) {0.\x};
            }
            \foreach \y in {0,1,...,9} {
                \draw(0,\y/10) -- (1,\y/10);                        \node [anchor=east] at (0,\y/10) {0.\y};
            }
        \end{scope}    
    },
    help grid/.style={
        append after command={
            [help grid code]
        }
    },
}

\newcommand\phantomimage{%
    \phantom{%
        \rule{\imagewidth}{\imageheight}%
    }%
}
\newcommand\zoombox[2][]{
    \begin{scope}[zoombox paths]
        \pgfmathsetmacro\xpos{
            (\columncount-1)*(\imagewidth / \pgfkeysvalueof{/tikz/zoomboxarray columns} + \pgfkeysvalueof{/tikz/zoomboxarray inner gap} / \pgfkeysvalueof{/tikz/zoomboxarray columns} ) + \pgflinewidth
        }
        \pgfmathsetmacro\ypos{
            (\rowcount-1)*( \imageheight / \pgfkeysvalueof{/tikz/zoomboxarray rows} + \pgfkeysvalueof{/tikz/zoomboxarray inner gap} / \pgfkeysvalueof{/tikz/zoomboxarray rows} ) + 0.5*\pgflinewidth
        }
        \edef\dospy{\noexpand\spy [
            #1,
            zoombox paths/.append style={
                black and white pattern=\patternnumber
            },
            every spy on node/.append style={#1},
            x=\imagewidth,
            y=\imageheight
        ] on (#2) in node [anchor=north west] at ($(zoomboxes container.north west)+(\xpos pt,-\ypos pt)$);}
        \dospy
        \pgfmathtruncatemacro\pgfmathresult{ifthenelse(\columncount==\pgfkeysvalueof{/tikz/zoomboxarray columns},\rowcount+1,\rowcount)}
        \global\let\rowcount=\pgfmathresult
        \pgfmathtruncatemacro\pgfmathresult{ifthenelse(\columncount==\pgfkeysvalueof{/tikz/zoomboxarray columns},1,\columncount+1)}
        \global\let\columncount=\pgfmathresult
        \ifblackandwhitecycle
            \pgfmathtruncatemacro{\newpatternnumber}{\patternnumber+1}
            \global\edef\patternnumber{\newpatternnumber}
        \fi
    \end{scope}
}

\usepackage{graphicx}
\usepackage{booktabs}
\usepackage{subcaption}
\usepackage{tikz}
\usepackage{xcolor}
\usepackage{rotating}
\usepackage{diagbox}
\usepackage{booktabs} 
\usepackage{multirow}
\usepackage[table]{xcolor}
\usepackage{stfloats}
\newcommand\blfootnote[1]{%
  \begingroup
  \renewcommand{\thefootnote}{}%
  \footnote{#1}%
  \addtocounter{footnote}{-1}
  \endgroup
}
\definecolor{cvprblue}{rgb}{0.21,0.49,0.74}
\usepackage[pagebackref,breaklinks,colorlinks,allcolors=cvprblue]{hyperref}
\title{Fractional-gradient Sparsity with Autoencoding Sequential Deep Image Prior for 3D CT Reconstruction}

\author{Haijie Yuan$^1$, Chaoyan Huang$^{1,2}$, Srijita Bandopadhyay$^1$, \\Liyue Shen$^2$, Saiprasad Ravishankar$^{1,3}$\\
$^{1}$Department of Computational Mathematics, Science, \& Engineering, Michigan State University  
$^{2}$Department of Electrical Engineering and Computer Science, University of Michigan\\
$^{3}$Department of Biomedical Engineering, Michigan State University \\
{\tt\small \{yuanhai1, huang345, bandopa1, ravisha3\}@msu.edu}\\
{\tt\small \{chaoyanh, liyues\}@umich.edu}
}

\begin{document}

\maketitle

\begin{abstract}
3D volumetric reconstruction from incomplete or noisy measurements is a fundamental problem in medical imaging and computational tomography. Deep image prior (DIP)–based methods have recently shown strong capability for solving inverse problems without requiring large training datasets. However, directly extending DIP to 3D reconstruction by fully 3D networks can incur high computational cost, while slice-by-slice 2D DIP approaches may lead to inter-slice inconsistencies due to the lack of explicit regularization along the third direction. 
In this paper, we propose a novel volumetric reconstruction framework, \textbf{F}ractional-gradient \textbf{A}utoencoding \textbf{S}equential \textbf{T}omography \textbf{DIP} (FAST-DIP), which integrates input-adaptive sequential deep image prior modeling of slices with fractional sparsity regularization to capture inter-slice dependencies.
Specifically, we introduce a fractional $\ell_1/\ell_2$-based sparsity prior on the gradients along the slice ($z$) direction to explicitly enforce inter-slice structural consistency. 
We further provide theoretical analysis of the proposed alternating minimization algorithm under the majorization–minimization (MM) framework, establishing monotonic descent of the objective function and convergence to a critical point under the Kurdyka–Łojasiewicz (KL) property.
Experimental results for 3D X-ray computed tomography (CT) reconstruction demonstrate that the proposed method improved reconstruction quality and structural consistency compared with existing DIP-based approaches.\blfootnote{DARPA Distribution Statement A. Approved for public release: distribution is
unlimited.}
\end{abstract}    
\section{Introduction}
\label{sec:intro}


X-ray computed tomography (CT) is widely used in medical, security, and scientific applications, and for nondestructive testing.
3D volumetric reconstruction from incomplete or noisy measurements is a central problem in 
CT~\cite{wang2023review}. The ill-posed nature of tasks such as sparse-view or limited-angle CT makes regularization essential for recovering diagnostically or physically meaningful volumes \cite{wang2021limited,cai2024structure}. In recent years, unsupervised methods based on the deep image prior (DIP) \cite{ulyanov2018deep} have demonstrated remarkable ability to regularize inverse problems without requiring large supervised datasets.
Obtaining large or well-matched datasets for training deep neural networks can be challenging in several scientific and health settings.

Several studies have extended the DIP framework to volumetric reconstruction tasks. For example, Gong et al. \cite{gong2018pet} incorporated DIP into an iterative reconstruction pipeline for 3D PET imaging. More recently, Hashimoto et al. \cite{hashimoto2023fully} implemented a fully 3D DIP-based network that directly parameterizes volumetric images. Other works further combine DIP with classical regularization strategies, such as sparsity-based priors \cite{han2023deep} or low-rank regularization \cite{chen2025multi}, in order to improve reconstruction quality and structural consistency. Although these approaches can achieve promising performance, fully 3D networks often suffer from high computational cost and memory consumption.

An alternative strategy is to apply 2D DIP slice-by-slice to reconstruct a 3D volume. For instance, Xue et al. \cite{xue2024hybrid} reconstructed 3D volumes by applying standard 2D DIP independently to each slice. However, purely slice-wise reconstruction may lead to inter-slice inconsistencies and lack explicit regularization along the slice direction. To alleviate this issue, Xue et al. \cite{xue2024hybrid} further introduced a 3D total variation (TV) regularization that enforces smoothness across slices.
In practice, medical images often contain rich structural details within each slice. Applying isotropic 3D regularization may therefore oversmooth fine intra-slice textures while enforcing volumetric consistency \cite{young2024fully}. 

DIP methods typically suffer from overfitting to noise/artifacts due to being trained directly with corrupted or limited measurements.
Recently, autoencoding sequential deep image prior (ASeqDIP) \cite{alkhouri2024image} was proposed that alternates between updating the DIP network parameters with an autoencoding regularization that delays overfitting and a gradient-free input update step that further denoises the network input (an image).
However, AseqDIP has not been applied to 3D reconstruction.
One approach is to apply it to 2D slices.
To impose structural sparsity along the third direction, several works impose regularization only along the slice axis to capture correlations between neighboring slices. 
In particular, z-axis TV regularization has been adopted to encourage inter-slice consistency in volumetric reconstruction~\cite{chung2023solving}. However, classical TV regularization is sensitive to the scale of gradients and may introduce staircase artifacts or overly smooth subtle inter-slice variations. These limitations motivate the development of priors that can more effectively model sparse inter-slice structures while remaining robust to gradient magnitude variations.

Classical model-based reconstruction methods incorporate explicit structural priors, such as Tikhonov regularization \cite{pan2020constrained} and nuclear norm regularization \cite{huang2023single}, to encode prior knowledge about the geometry of volumetric data. Inspired by these observations, we introduce a fractional sparsity prior based on the $\ell_1/\ell_2$ ratio. As shown in \cite{wang2021limited}, the $\ell_1/\ell_2$ fraction-based regularization promotes structured sparsity by encouraging many gradients to approach zero while avoiding the trivial shrinkage of the gradient field that may occur with standard $\ell_1$ penalties.

In this paper, we propose a novel 3D volumetric reconstruction framework, \textbf{F}ractional-gradient \textbf{A}utoencoding \textbf{S}equential \textbf{T}omography \textbf{DIP} (FAST-DIP), that integrates 2D deep image prior with fractional sparsity regularization.
Specifically, we leverage the sequential input-adaptive modeling strategy of ASeqDIP over 2D slices.
In addition, to explicitly enforce inter-slice structural consistency, we introduce an $\ell_1/\ell_2$ regularization on the gradients along the slice direction.
The main contributions of this work are summarized as follows:

\begin{itemize}

\item We introduce an $\ell_1/\ell_2$-based fractional sparsity regularization to model sparse variations along the slice direction. The proposed prior is less sensitive to gradient magnitude and better preserves subtle inter-slice structural variations.

\item We incorporate the proposed fractional sparsity prior into a DIP-based reconstruction framework with sequential slice modeling, enabling the network to reconstruct slices while explicitly enforcing structural consistency across neighboring slices.


\item We provide theoretical guarantees for the proposed alternating optimization algorithm under the majorization–minimization (MM) framework. In particular, we establish monotonic descent of the objective function and prove subsequence convergence to a critical point based on the Kurdyka–Łojasiewicz (KL) property.

\item Experiments for challenging 3D CT reconstruction tasks demonstrate that the proposed method achieves improved reconstruction quality compared with other DIP-based approaches.

\end{itemize}
\section{Related Works}
\label{sec:related}

\subsection{ASeqDIP}

The Autoencoding Sequential Deep Image Prior (ASeqDIP) \cite{alkhouri2024image} extends the classical Deep Image Prior (DIP) \cite{ulyanov2018deep} framework by incorporating sequential input adaptation during reconstruction. Consider an inverse imaging problem where the observation $\textbf{y} \in \mathbb{R}^m$ is generated through a forward operator $\textbf{A}$ as:
\begin{equation}
\textbf{y} = \textbf{A} \textbf{x} + \textbf{n},
\end{equation}
where $\textbf{x} \in \mathbb{R}^n$ (discretized) denotes the unknown clean image and $\textbf{n}$ represents measurement noise. 
In the original DIP formulation, the reconstruction is produced by a randomly initialized neural network $f_{\phi}(\cdot)$ that maps a fixed random input $\textbf{z}$ to an image estimate. The network parameters are optimized by minimizing a data consistency objective:
\begin{equation}
\min_{\phi} \; \| \mathbf{A}f_{\phi}\mathbf{(z)} - \mathbf{y} \|_2^2.
\end{equation}

Recent work has shown that while DIP provides a strong implicit prior, its optimization can become unstable and overfit measurement noise \cite{liang2025analysis}, motivating additional regularization and stabilization strategies \cite{cheng2019bayesian,wang2023early}. 
ASeqDIP \cite{alkhouri2024image} used an autoencoding constraint together with sequential input refinement to improve the DIP method.
It introduced a sequence of progressively refined inputs. In stage $t$, the network generates an intermediate reconstruction $f_{\phi}(\mathbf{z}_t)$ 
and the optimization objective becomes:
\begin{equation}
\min_{\phi}
\|\textbf{A} f_{\phi}(\textbf{z}_t)- \textbf{y}\|_2^2
+
\lambda \|f_{\phi}(\textbf{z}_t) - \textbf{z}_t\|_2^2.
\end{equation}
The additional autoencoding regularization term encourages the output to remain close to its input, preventing unstable updates while guiding the reconstruction toward measurement-consistent structures.
After each stage, the output (with refined network) is used to update the next input as:
\begin{equation}
\mathbf{z}_{t+1} = \textbf{x}_t = f_{\phi}(\mathbf{z}_t). 
\end{equation}
Changing between network optimization and input refinement, ASeqDIP produces a sequence of progressively improved or denoised reconstructions. This sequential strategy helps mitigate the overfitting behavior commonly observed in standard DIP optimization \cite{wang2023early}.

\subsection{Fractional-Gradient Sparsity Prior}
Sparsity-promoting regularization plays an important role in inverse imaging problems. Classical approaches such as $\ell_1$ regularization and total variation (TV) enforce the sparsity of image gradients and have been widely used for image reconstruction tasks \cite{ritschl2011improved,huang2024edge,wu2024medical}. In particular, the TV seminorm can be interpreted as the $\ell_1$ norm of the image gradient \cite{wang2022minimizing}. It is well known that the $\ell_1$ norm serves as the tightest convex relaxation of the $\ell_0$ norm, which directly measures sparsity of signals. However, TV-based regularization is sensitive to the scale of gradients and may introduce staircase artifacts when modeling piecewise-smooth structures \cite{jiang2023inexact}.

To better approximate the $\ell_0$ norm, several alternative sparsity-promoting regularizations have been proposed, including nonconvex $\ell_p$ penalties with $0<p<1$ \cite{wu2021two} and fractional sparsity formulations such as the $\ell_1/\ell_2$ ratio \cite{wang2022minimizing}. In particular, the $\ell_1/\ell_2$ regularization has been shown to effectively mitigate the staircasing artifacts commonly produced by TV, since the $\ell_2$ norm of the gradient in the denominator discourages trivial shrinkage of the entire gradient field \cite{wang2021limited}. Moreover, the $\ell_1/\ell_2$ penalty promotes structured sparsity by encouraging many gradient components to approach zero while preserving the relative magnitude of significant gradients. Due to its scale-invariant property, this type of regularization has demonstrated advantages in various inverse problems, including compressed sensing and image restoration.

Motivated by these observations, we adopt a fractional-gradient sparsity prior to model sparse variations along the slice direction in volumetric reconstruction. The proposed regularization focuses on the sparsity of inter-slice gradients, which helps preserve intra-slice structural details while enforcing structural consistency across neighboring slices.

\section{Method}
\label{sec:method}
The overall objective function of the proposed FAST-DIP is defined as follows:
\begin{equation}
    \min_{\phi,\mathbf{z}} \mathcal{L}(\phi, \mathbf{z}) = F(\phi, \mathbf{z}) + G(\mathbf{z}),
\end{equation}
where $F(\phi, \mathbf{z})$ encompasses the data fidelity and the autoencoding sequential prior, and $G(\mathbf{z})$ is the fractional gradient sparsity prior:
\begin{subequations}
\begin{equation}
    F(\phi, \mathbf{z}) = \left\|\mathbf{A} f_{\phi}(\mathbf{z})-\mathbf{y}\right\|_2^2 + \lambda\Vert\mathbf{z}-f_{\phi}(\mathbf{z})\Vert_2^2,
\end{equation}
\begin{equation}
    G(\mathbf{z}) = \gamma\frac{\Vert\nabla \mathbf{z}\Vert_1}{\Vert\nabla \mathbf{z}\Vert_2}, 
\end{equation}
\end{subequations}
where $f_{\phi}(\cdot)$ acts on each 2D x-y slice of a 3D volume and $\nabla$ is the gradient along the $z$-axis. 
In contrast to ASeqDIP that avoided gradient-based input updates and used a simple forward pass to update the input, here we optimize over the input $\mathbf{z}$ as well in the loss function and gradients will be used for its update given the additional regularization.

To address the non-differentiability of the $\ell_1$ norm at zero and the potential zero-denominator in the $\ell_2$ norm, we introduce smoothing parameters $\epsilon, \delta > 0$ to define a smoothed regularization term: 
\begin{equation}
    G_{\epsilon, \delta}(\mathbf{z}) = \gamma \frac{\sum_i \sqrt{|(\nabla \mathbf{z})_i|^2 + \delta}}{\|\nabla \mathbf{z}\|_2 + \epsilon}.
\end{equation}
Due to the concavity of the square root function, the numerator can be upper-bounded by a quadratic function via the first-order Taylor expansion. By concurrently fixing the denominator at a current iterate $\mathbf{z}^{k-1}$, the explicit form of the quadratic surrogate function $S(\mathbf{z} | \mathbf{z}^{k-1})$ is formulated as:
\begin{equation}\label{seq}
    S(\mathbf{z} | \mathbf{z}^{k-1}) = \frac{\gamma}{2 M^{k-1}} (\nabla \mathbf{z})^T \mathbf{W}^{k-1} (\nabla \mathbf{z}) + C,
\end{equation}
where $C$ is a constant independent of $\mathbf{z}$, the scalar $M^{k-1} = \|\nabla \mathbf{z}^{k-1}\|_2 + \epsilon$ freezes the denominator, and $\mathbf{W}^{k-1}$ is a diagonal weight matrix with entries $\mathbf{W}_{i,i}^{k-1} = 1 / \sqrt{|(\nabla \mathbf{z}^{k-1})_i|^2 + \delta}$. 

Instead of utilizing proximal operators which are intractable for the fractional term, we employ an alternating gradient descent approach based on the Majorization-Minimization (MM) framework to solve the smoothed objective $\mathcal{L}_{\epsilon, \delta}(\phi, \mathbf{z})$. 

At the $k$-th iteration, the alternating update rules are formulated as:
\begin{equation}
    \begin{aligned}
        \phi^k &= \phi^{k-1} - \alpha\nabla_{\phi}F(\phi^{k-1},\mathbf{z}^{k-1}),\\
        \mathbf{z}^k &= \mathbf{z}^{k-1} - \beta \nabla_{\mathbf{z}} \mathcal{L}_{s}(\phi^k, \mathbf{z} | \mathbf{z}^{k-1}),
    \end{aligned}
\end{equation}
where $\alpha$ and $\beta$ are the stepsizes, and $\mathcal{L}_{s}(\phi, \mathbf{z} | \mathbf{z}^{k-1}) = F(\phi, \mathbf{z}) + S(\mathbf{z} | \mathbf{z}^{k-1})$ is the surrogate objective function evaluated at the current point $\mathbf{z}^{k-1}$. The overall procedure is summarized in Algorithm \ref{algo:mm_agd}.

\begin{algorithm*}[h]
\caption{MM-based Alternating Gradient Descent Algorithm}
\label{algo:mm_agd}
\SetKwInOut{Input}{Input}
\SetKwInOut{Output}{Output}
\SetKwComment{Comment}{$\triangleright$\ }{}

\Input{~Measurements $\mathbf{y}$, forward operator $\mathbf{A}$, neural network model $f_{\phi}$.}
\Input{~Regularization parameters $\lambda, \gamma > 0$, smoothing parameters $\epsilon, \delta > 0$ (e.g., $10^{-6}$).}
\Input{~Stepsizes $\alpha$ (for $\phi$) and $\beta$ (for $\mathbf{z}$).}
\Input{~Initialization $\phi^0$ and $\mathbf{z}^0$.}

$k \leftarrow 1$\;
\While{stopping criterion is not met}{
    \vspace{0.1cm}
    \Comment{Step 1: Update network parameters $\phi$}
    Compute the gradient with respect to $\phi$:\\
    $\nabla_{\phi} F \leftarrow \nabla_{\phi} \left( \|\mathbf{A} f_{\phi^{k-1}}(\mathbf{z}^{k-1}) - \mathbf{y}\|_2^2 + \lambda \|\mathbf{z}^{k-1} - f_{\phi^{k-1}}(\mathbf{z}^{k-1})\|_2^2 \right)$\;
    $\phi^k \leftarrow \phi^{k-1} - \alpha \nabla_{\phi} F$\;
    
    \vspace{0.1cm}
    \Comment{Step 2: Construct the quadratic surrogate for the fractional term}
    Compute the spatial gradients of the current image:\\
    $\mathbf{g}^{k-1} \leftarrow \nabla \mathbf{z}^{k-1}$\;
    Fix the $L_2$ norm denominator as a constant scalar:\\
    $M^{k-1} \leftarrow \|\mathbf{g}^{k-1}\|_2 + \epsilon$\;
    Construct the diagonal reweighting matrix $\mathbf{W}^{k-1}$ with diagonal entries:\\
    $\mathbf{W}_{i,i}^{k-1} \leftarrow \frac{1}{\sqrt{|\mathbf{g}_i^{k-1}|^2 + \delta}}$\;
    
    \vspace{0.1cm}
    \Comment{Step 3: Update image variable $\mathbf{z}$}
    Compute the gradient of the surrogate fractional term:\\
    $\mathbf{v}^{k-1} \leftarrow \frac{\gamma}{M^{k-1}} \nabla^T \left( \mathbf{W}^{k-1} \nabla \mathbf{z}^{k-1} \right)$\;
    Compute the gradient of the smooth data fidelity term:\\
    $\nabla_{\mathbf{z}} F \leftarrow \nabla_{\mathbf{z}} \left( \|\mathbf{A} f_{\phi^k}(\mathbf{z}) - \mathbf{y}\|_2^2 + \lambda \|\mathbf{z} - f_{\phi^k}(\mathbf{z})\|_2^2 \right) \Big|_{\mathbf{z}=\mathbf{z}^{k-1}}$\;
    Perform gradient descent on the surrogate objective:\\
    $\mathbf{z}^k \leftarrow \mathbf{z}^{k-1} - \beta \left( \nabla_{\mathbf{z}} F + \mathbf{v}^{k-1} \right)$\;
    
    $k \leftarrow k + 1$\;
}
\Output{Optimized network parameters $\phi^*$ and reconstructed image $\mathbf{z}^*$.}
\end{algorithm*}

\begin{remark}[Relation to Weighted Quadratic Regularization]
At each iteration, the surrogate function $S(\mathbf{z} \mid \mathbf{z}^{k-1})$ takes the form of a weighted quadratic gradient penalty \eqref{seq}, which resembles a Tikhonov regularization with spatially varying weights. 
However, it is important to emphasize that the weights $W^{k-1}$ are \emph{iteratively updated} based on the previous iterate $\mathbf{z}^{k-1}$, rather than being fixed. As a result, the overall optimization process does not correspond to a static quadratic regularization. Instead, it follows an iteratively reweighted least squares (IRLS) scheme, which is well known to approximate non-quadratic sparsity-inducing penalties.
\end{remark}
\section{Theoretical Analysis}
\label{sec:theoral}
In this section, we provide theoretical guarantees for the proposed alternating minimization algorithm under the Majorization-Minimization (MM) framework.

\begin{lemma}[Majorization of the Smoothed Numerator]\label{lemma1}
Let $
N_{\delta}(\mathbf{z}) = \sum_i \sqrt{|(\nabla \mathbf{z})_i|^2 + \delta}$, and for any $\mathbf{z}$ and fixed $\mathbf{z}^{k-1}$, the constructed quadratic surrogate $\tilde{N}_{\delta}(\mathbf{z} | \mathbf{z}^{k-1})$ satisfies the tangency and upper-bound conditions:
\begin{equation}
\tilde{N}_{\delta}(\mathbf{z}^{k-1} | \mathbf{z}^{k-1}) = N_{\delta}(\mathbf{z}^{k-1}), 
~~
\tilde{N}_{\delta}(\mathbf{z} | \mathbf{z}^{k-1}) \ge N_{\delta}(\mathbf{z}), \ \forall \mathbf{z}.
\end{equation}

\end{lemma}
\begin{remark}
Using the frozen denominator
\begin{equation}
M^{k-1} = \|\nabla \mathbf{z}^{k-1}\|_2 + \epsilon,
\end{equation}
the practical surrogate used in our algorithm is
\begin{equation}
S(\mathbf{z}\mid \mathbf{z}^{k-1})
=
\frac{\gamma}{M^{k-1}}
\tilde N_{\delta}(\mathbf{z}\mid \mathbf{z}^{k-1}).
\end{equation}
We emphasize that this construction majorizes the smoothed numerator, while the denominator is frozen at the previous iterate. Therefore, $S(\mathbf{z}\mid \mathbf{z}^{k-1})$ should be interpreted as an IRLS-type surrogate for the original fractional prior, rather than a global majorizer of the full ratio functional $G_{\epsilon,\delta}(\mathbf{z})$.
\end{remark}

\begin{lemma}[Lipschitz Continuous Gradient]\label{lemma2}
Assume that for any fixed $\phi^k$, the mapping $f_{\phi^k}(\mathbf z)$ is continuously differentiable and has bounded Jacobian and Hessian on the iterate sequence.
Then the gradient of the surrogate objective with respect to $\mathbf z$, defined as
\[
\nabla_{\mathbf z} \mathcal{L}_s(\phi^k, \mathbf z \mid \mathbf z^{k-1}),
\]
is $L$-Lipschitz continuous on the bounded domain of iterates. That is, there exists a constant $L > 0$ such that for any $\mathbf z_1, \mathbf z_2$:
\begin{equation}
\|\nabla_{\mathbf z} \mathcal{L}_s(\phi^k, \mathbf z_1 \mid \mathbf z^{k-1})
-
\nabla_{\mathbf z} \mathcal{L}_s(\phi^k, \mathbf z_2 \mid \mathbf z^{k-1})\|_2
\le
L \|\mathbf z_1 - \mathbf z_2\|_2.
\end{equation}
\end{lemma}

\begin{proposition}[Monotonic Descent of the Surrogate Objective]\label{proposition1}
Let $\beta \le \frac{1}{L}$, where $L$ is the Lipschitz constant in Lemma \ref{lemma2}. 
Then the update
\[
\mathbf z^k
=
\mathbf z^{k-1}
-
\beta \nabla_{\mathbf z}
\mathcal{L}_s(\phi^k, \mathbf z \mid \mathbf z^{k-1})
\Big|_{\mathbf z = \mathbf z^{k-1}}
\]
guarantees a monotonic decrease of the surrogate objective:
\begin{equation}
\mathcal{L}_s(\phi^k, \mathbf z^k \mid \mathbf z^{k-1})
\le
\mathcal{L}_s(\phi^k, \mathbf z^{k-1} \mid \mathbf z^{k-1}).
\end{equation}
\end{proposition}

\begin{theorem}[Subsequence Convergence under KL Framework]\label{theorem1}
Assume that the sequence $\{(\phi^k, \mathbf z^k)\}_{k=1}^\infty$ generated by the alternating updates is bounded, and that the surrogate-based updates satisfy a sufficient decrease condition and a relative error condition.
Since the data fidelity term, the autoencoding regularization, and the surrogate fractional prior are semi-algebraic functions, the overall objective satisfies the Kurdyka--{\L}ojasiewicz (KL) property.

Then the sequence $\{(\phi^k, \mathbf z^k)\}_{k=1}^\infty$ has finite length, and any accumulation point is a critical point of the surrogate-based alternating scheme.
\end{theorem}

\begin{remark}
   This convergence statement is consistent with classical IRLS analyses \cite{mohan2012iterative}, where each subproblem is quadratic, while the overall algorithm targets a nonquadratic objective through iterative reweighting. In such settings, one typically proves that the difference between successive iterates vanishes and that every cluster point is stationary. 
\end{remark}

Due to space constraints, the rigorous mathematical proofs for aforementioned theoretical claims are deferred to the supplementary material.
Specifically, we establish a valid majorization bound for the smoothed numerator via the first-order Taylor expansion of concave functions (Lemma \ref{lemma1}), prove monotonic descent of the surrogate objective based on the $L$-Lipschitz continuity of the surrogate gradient (Lemma \ref{lemma2} and Proposition \ref{proposition1}), and analyze the subsequence convergence behavior of the surrogate-based alternating scheme under the Kurdyka--{\L}ojasiewicz (KL) framework (Theorem \ref{theorem1}).

\section{Experiments}
\label{sec:exp}

\begin{figure}[t]
\setlength{\abovecaptionskip}{0.05in}
    \centering
\begin{subfigure}[t]{0.19\textwidth}
\begin{tikzpicture}[zoomboxarray,zoomboxes yshift=1.21,zoomboxes xshift=0.,spymargin=2pt]
\pgfkeys{/tikz/zoomboxarray columns=2,/tikz/zoomboxarray rows=1,/tikz/zoomboxarray inner gap=1}
\node [image node] {
\includegraphics[width=\textwidth]{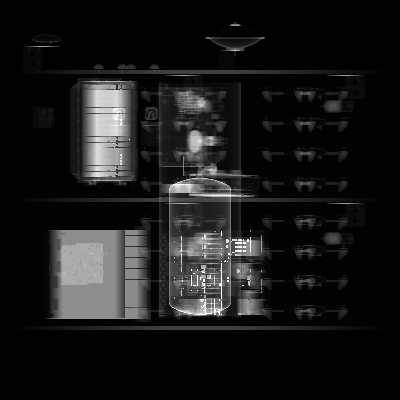}
};
\zoombox[magnification=3,color code=yellow,zoombox paths/.append style={line width=1pt}]{0.49,0.52}
\gdef\rowcount{2}
\zoombox[magnification=3,color code=orange,zoombox paths/.append style={line width=1pt}]{0.33,0.34}
\end{tikzpicture}
\centering \footnotesize (a) GT ($\infty/1$)
\end{subfigure}
\begin{subfigure}[t]{0.19\textwidth}
\begin{tikzpicture}[zoomboxarray,zoomboxes yshift=1.21,zoomboxes xshift=0.,spymargin=2pt]
\pgfkeys{/tikz/zoomboxarray columns=2,/tikz/zoomboxarray rows=1,/tikz/zoomboxarray inner gap=1}
\node [image node] {
\includegraphics[width=\textwidth]{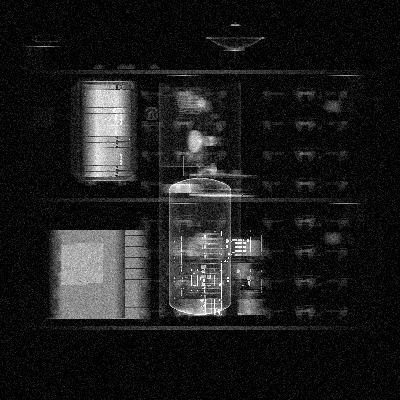}
};
\zoombox[magnification=3,color code=yellow,zoombox paths/.append style={line width=1pt}]{0.49,0.52}
\gdef\rowcount{2}
\zoombox[magnification=3,color code=orange,zoombox paths/.append style={line width=1pt}]{0.33,0.34}
\end{tikzpicture}
\centering \footnotesize (b) Noisy (26.51/0.396)
\end{subfigure}
\begin{subfigure}[t]{0.19\textwidth}
\begin{tikzpicture}[zoomboxarray,zoomboxes yshift=1.21,zoomboxes xshift=0.,spymargin=2pt]
\pgfkeys{/tikz/zoomboxarray columns=2,/tikz/zoomboxarray rows=1,/tikz/zoomboxarray inner gap=1}
\node [image node] {
\includegraphics[width=\textwidth]{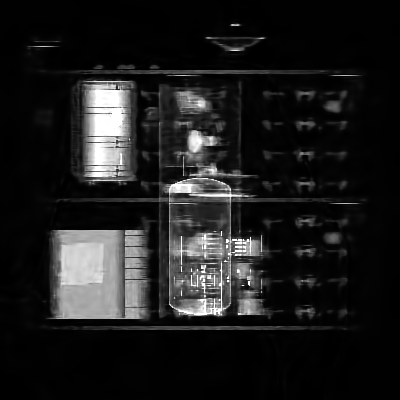}
};
\zoombox[magnification=3,color code=yellow,zoombox paths/.append style={line width=1pt}]{0.49,0.52}
\gdef\rowcount{2}
\zoombox[magnification=3,color code=orange,zoombox paths/.append style={line width=1pt}]{0.33,0.34}
\end{tikzpicture}
\centering \footnotesize (c) ASeqDIP (28.94/0.554)
\end{subfigure}
\begin{subfigure}[t]{0.19\textwidth}
\begin{tikzpicture}[zoomboxarray,zoomboxes yshift=1.21,zoomboxes xshift=0.,spymargin=2pt]
\pgfkeys{/tikz/zoomboxarray columns=2,/tikz/zoomboxarray rows=1,/tikz/zoomboxarray inner gap=1}
\node [image node] {
\includegraphics[width=\textwidth]{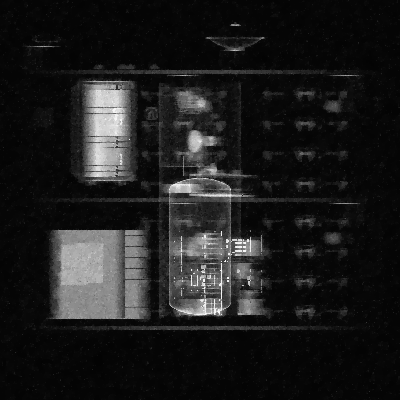}
};
\zoombox[magnification=3,color code=yellow,zoombox paths/.append style={line width=1pt}]{0.49,0.52}
\gdef\rowcount{2}
\zoombox[magnification=3,color code=orange,zoombox paths/.append style={line width=1pt}]{0.33,0.34}
\end{tikzpicture}
\centering \footnotesize (d) $\ell_1/\ell_2$ (30.96/0.554)
\end{subfigure}
\begin{subfigure}[t]{0.19\textwidth}
\begin{tikzpicture}[zoomboxarray,zoomboxes yshift=1.21,zoomboxes xshift=0.,spymargin=2pt]
\pgfkeys{/tikz/zoomboxarray columns=2,/tikz/zoomboxarray rows=1,/tikz/zoomboxarray inner gap=1}
\node [image node] {
\includegraphics[width=\textwidth]{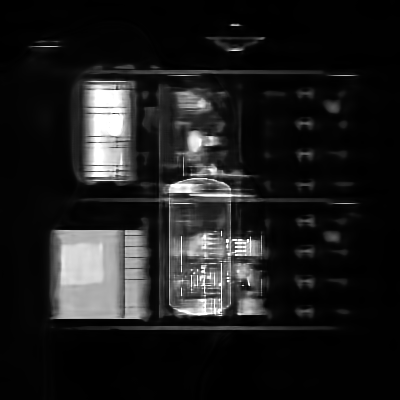}
};
\zoombox[magnification=3,color code=yellow,zoombox paths/.append style={line width=1pt}]{0.49,0.52}
\gdef\rowcount{2}
\zoombox[magnification=3,color code=orange,zoombox paths/.append style={line width=1pt}]{0.33,0.34}
\end{tikzpicture}
\centering \footnotesize (e) FAST-DIP (30.56/0.582)
\end{subfigure}  

\vspace{-1.1in}
\begin{subfigure}[t]{0.19\textwidth}
\begin{tikzpicture}[zoomboxarray,zoomboxes yshift=1.21,zoomboxes xshift=0.,spymargin=2pt]
\pgfkeys{/tikz/zoomboxarray columns=2,/tikz/zoomboxarray rows=1,/tikz/zoomboxarray inner gap=1}
\node [image node] {
\includegraphics[width=\textwidth]{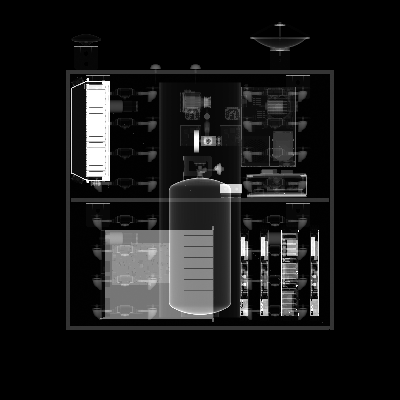}
};
\zoombox[magnification=3.5,color code=yellow,zoombox paths/.append style={line width=1pt}]{0.52,0.36}
\gdef\rowcount{2}
\zoombox[magnification=5,color code=orange,zoombox paths/.append style={line width=1pt}]{0.245,0.56}
\end{tikzpicture}
\centering \footnotesize (a) GT ($\infty$/1)
\end{subfigure} 
\begin{subfigure}[t]{0.19\textwidth}
\begin{tikzpicture}[zoomboxarray,zoomboxes yshift=1.21,zoomboxes xshift=0.,spymargin=2pt]
\pgfkeys{/tikz/zoomboxarray columns=2,/tikz/zoomboxarray rows=1,/tikz/zoomboxarray inner gap=1}
\node [image node] {
\includegraphics[width=\textwidth]{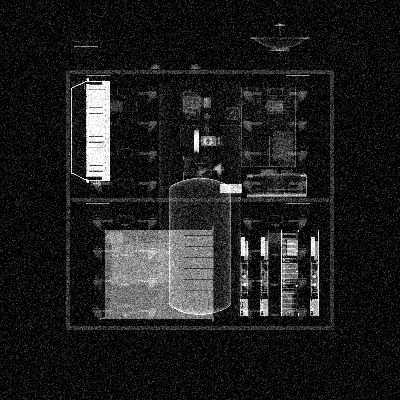}
};
\zoombox[magnification=3.5,color code=yellow,zoombox paths/.append style={line width=1pt}]{0.52,0.36}
\gdef\rowcount{2}
\zoombox[magnification=5,color code=orange,zoombox paths/.append style={line width=1pt}]{0.245,0.56}
\end{tikzpicture}
\centering \footnotesize (b) Noisy (22.36/0.261) 
\end{subfigure}
\begin{subfigure}[t]{0.19\textwidth}
\begin{tikzpicture}[zoomboxarray,zoomboxes yshift=1.21,zoomboxes xshift=0.,spymargin=2pt]
\pgfkeys{/tikz/zoomboxarray columns=2,/tikz/zoomboxarray rows=1,/tikz/zoomboxarray inner gap=1}
\node [image node] {
\includegraphics[width=\textwidth]{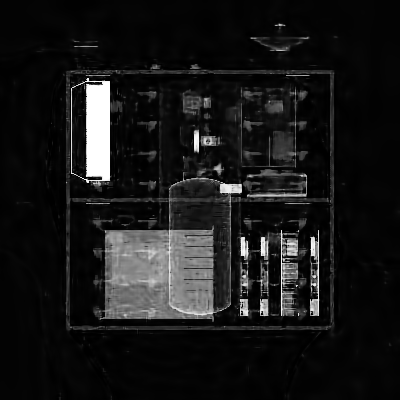}
};
\zoombox[magnification=3.5,color code=yellow,zoombox paths/.append style={line width=1pt}]{0.52,0.36}
\gdef\rowcount{2}
\zoombox[magnification=5,color code=orange,zoombox paths/.append style={line width=1pt}]{0.245,0.56}
\end{tikzpicture}
\centering \footnotesize (c) ASeqDIP (26.44/0.399)
\end{subfigure}
\begin{subfigure}[t]{0.19\textwidth}
\begin{tikzpicture}[zoomboxarray,zoomboxes yshift=1.21,zoomboxes xshift=0.,spymargin=2pt]
\pgfkeys{/tikz/zoomboxarray columns=2,/tikz/zoomboxarray rows=1,/tikz/zoomboxarray inner gap=1}
\node [image node] {
\includegraphics[width=\textwidth]{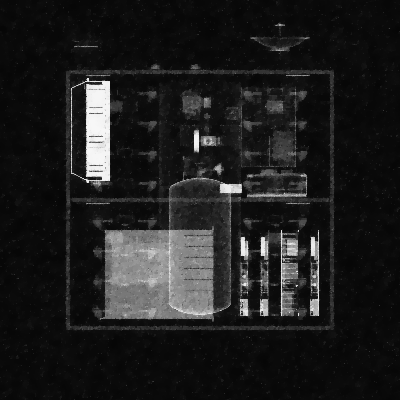}
};
\zoombox[magnification=3.5,color code=yellow,zoombox paths/.append style={line width=1pt}]{0.52,0.36}
\gdef\rowcount{2}
\zoombox[magnification=5,color code=orange,zoombox paths/.append style={line width=1pt}]{0.245,0.56}
\end{tikzpicture}
\centering \footnotesize (d) $\ell_1/\ell_2$ (26.87/0.383)
\end{subfigure}
\begin{subfigure}[t]{0.19\textwidth}
\begin{tikzpicture}[zoomboxarray,zoomboxes yshift=1.21,zoomboxes xshift=0.,spymargin=2pt]
\pgfkeys{/tikz/zoomboxarray columns=2,/tikz/zoomboxarray rows=1,/tikz/zoomboxarray inner gap=1}
\node [image node] {
\includegraphics[width=\textwidth]{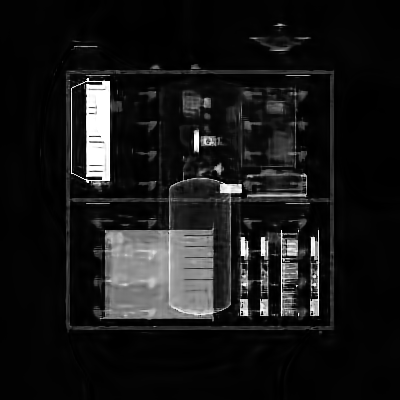}
};
\zoombox[magnification=3.5,color code=yellow,zoombox paths/.append style={line width=1pt}]{0.52,0.36}
\gdef\rowcount{2}
\zoombox[magnification=5,color code=orange,zoombox paths/.append style={line width=1pt}]{0.245,0.56}
\end{tikzpicture}
\centering \footnotesize (e) FAST-DIP (27.06/0.411)
\end{subfigure}
\caption{Visual comparisons of CT data denoising with Gaussian noise. The first row is with noise level $\sigma=15$, and the second row is with noise level $\sigma=25$. From left to right: ground truth (GT), noisy input, ASeqDIP \cite{alkhouri2024image}, traditional $\ell_1/\ell_2$ method \cite{wang2022minimizing}, and the proposed FAST-DIP.}
    \label{fig: denoising25}
\end{figure}

\subsection{Implementation}
To rigorously assess the effectiveness of the proposed FAST-DIP reconstruction model, we perform a series of experiments on image denoising and three-dimensional (3D) reconstruction using a synthetic dataset derived from XENA \cite{xenadarpawebsite}. 
{The dataset is a simulated CT scenario, in which a drone-mounted sensor rotates around a box-shaped target with objects placed inside at a standoff distance of 1000 m. Projection data are collected under a parallel-beam geometry spanning a full 360° angular range, with each projection individually stored as a TIFF file indexed by its angle. It is restricted and cannot be made publicly available due to confidentiality constraints.}
All experiments are implemented in PyTorch and conducted on an NVIDIA RTX PRO 6000 Blackwell GPU.
For fair comparison, we adopt the same network architecture as ASeqDIP, utilizing an encoder-decoder skip network with reflection padding and LeakyReLU activations. 

In our method, the network takes intermediate variable $\mathbf{z}$ (initialized as a random variable) as input. The parameters $\phi$ are optimized using the Adam optimizer with learning rate $1e-4$, while the projection $\mathbf{z}$ is updated via gradient descent. The alternating gradient descent process runs for $K = 800$ outer iterations, with the network updated for $N = 2$ inner steps per iteration. The autoencoding loss weight is set to $\lambda = 1.0$, the fractional gradient penalty weight is set to $\gamma = 0.01$, and the step size is $\beta = 0.9$. 

\begin{table}[t]
\centering
\caption{
Average quantitative comparison of image denoising on 360 simulated 2D images from the XENA dataset~\cite{xenadarpawebsite} under Gaussian noise.
The best result is shown in bold, and the second-best result is underlined.
}
\label{tab:2Ddenoising_results}
\small
\setlength{\tabcolsep}{7pt}
\renewcommand{\arraystretch}{1.15}
\begin{tabular}{lcccc}
\toprule
\multirow{2}{*}{\diagbox[width=8.5em,height=2.4\line]{Method}{Noise}}
& \multicolumn{2}{c}{$\sigma=15$}
& \multicolumn{2}{c}{$\sigma=25$} \\
\cmidrule(lr){2-3}
\cmidrule(lr){4-5}
& PSNR $\uparrow$ & SSIM $\uparrow$
& PSNR $\uparrow$ & SSIM $\uparrow$ \\
\midrule
Noisy
& 26.57 & 0.389
& 22.30 & 0.255 \\

ASeqDIP
& 30.22 & \underline{0.560}
& 26.83 & \underline{0.431} \\

$\ell_1/\ell_2$
& \textbf{30.90} & 0.530
& \textbf{27.25} & 0.425 \\

\textbf{FAST-DIP}
& \underline{30.66} & \textbf{0.566}
& \underline{27.20} & \textbf{0.456} \\
\bottomrule
\end{tabular}
\end{table}

\subsection{2D Denoising}
To evaluate the proposed FAST-DIP, we first implement a denoising task. We set the forward operator $\mathbf{A}=\mathbf{I}$ in our model. We compare ASeqDIP \cite{alkhouri2024image}, traditional $\ell_1/\ell_2$ \cite{wang2022minimizing}, and our method for images from XENA \cite{xenadarpawebsite} dataset corrupted with Gaussian noise ($\sigma=15$ and $\sigma=25$). 
The average quantitative results are presented in Table \ref{tab:2Ddenoising_results}. Regarding the ASeqDIP method, we followed the original paper's specifications and executed 2,000 iterations to obtain the optimal result. However, FAST-DIP required only 800 iterations to achieve the optimal outcome. Moreover, our method yields superior denoising performance compared to ASeqDIP. As indicated in Table \ref{tab:2Ddenoising_results}, our results are approximately 0.4 to 0.5 dB higher than those of ASeqDIP on average. Furthermore, we conducted a comparison with the traditional $\ell_1/\ell_2$ method \cite{wang2022minimizing}, which further highlights the advantages of our approach.

\begin{table*}[t]
\centering
\caption{
Average quantitative comparison of 3D reconstruction on simulated 20 given views and 100 novel views from the XENA dataset~\cite{xenadarpawebsite}.
The best result is shown in bold, and the second-best result is underlined.
}
\label{tab:3Dreconstruction-results}
\small
\setlength{\tabcolsep}{4.8pt}
\renewcommand{\arraystretch}{1.15}
\begin{tabular}{lcccccccc}
\toprule
\multirow{3}{*}{\diagbox[width=8.5em,height=3.2\line]{Method}{Setting}}
& \multicolumn{4}{c}{Sparse Views}
& \multicolumn{4}{c}{Limited Angles} \\
\cmidrule(lr){2-5}
\cmidrule(lr){6-9}
& \multicolumn{2}{c}{Given Views}
& \multicolumn{2}{c}{Novel Views}
& \multicolumn{2}{c}{Given Views}
& \multicolumn{2}{c}{Novel Views} \\
\cmidrule(lr){2-3}
\cmidrule(lr){4-5}
\cmidrule(lr){6-7}
\cmidrule(lr){8-9}
& PSNR $\uparrow$ & MS-SSIM $\uparrow$
& PSNR $\uparrow$ & MS-SSIM $\uparrow$
& PSNR $\uparrow$ & MS-SSIM $\uparrow$
& PSNR $\uparrow$ & MS-SSIM $\uparrow$ \\
\midrule
FBP
& 30.44 & 0.956
& 28.78 & 0.935
& 22.47 & 0.776
& 20.50 & 0.655 \\

ASeqDIP
& 40.61 & \underline{0.991}
& 31.91 & 0.970
& 28.21 & 0.921
& 22.27 & 0.731 \\

ASeqDIP + TV
& \underline{42.81} & \textbf{0.996}
& \underline{32.49} & \underline{0.973}
& \underline{42.73} & \underline{0.996}
& \underline{27.30} & \textbf{0.865} \\

\textbf{FAST-DIP}
& \textbf{43.33} & \textbf{0.996}
& \textbf{32.92} & \textbf{0.976}
& \textbf{45.11} & \textbf{0.999}
& \textbf{27.46} & \textbf{0.865} \\
\bottomrule
\end{tabular}
\end{table*}

Additionally, we present the visual results of Gaussian noise level $\sigma=15$ and $\sigma=25$ in Fig. \ref{fig: denoising25}. 
ASeqDIP suppresses noise but tends to oversmooth certain regions, reducing structural fidelity. The $\ell_1/\ell_2$ method reduces some noise but still leaves visible artifacts and loses fine details. 
In contrast, our method produces cleaner images while preserving sharper edges and structural details. As highlighted in the zoomed regions, our approach better recovers thin boundaries and texture patterns that remain noisy in the other methods. 
Overall, the visual results show that our method achieves a better balance between noise removal and detail preservation, which is consistent with the quantitative improvements reported in Table \ref{tab:2Ddenoising_results}.

\begin{figure*}[t]
\setlength{\abovecaptionskip}{0.05in}
    \centering
\begin{subfigure}[t]{0.19\textwidth}
\begin{tikzpicture}[zoomboxarray,zoomboxes yshift=1.21,zoomboxes xshift=0.,spymargin=2pt]
\pgfkeys{/tikz/zoomboxarray columns=2,/tikz/zoomboxarray rows=1,/tikz/zoomboxarray inner gap=1}
\node [image node] {
\includegraphics[width=\textwidth]{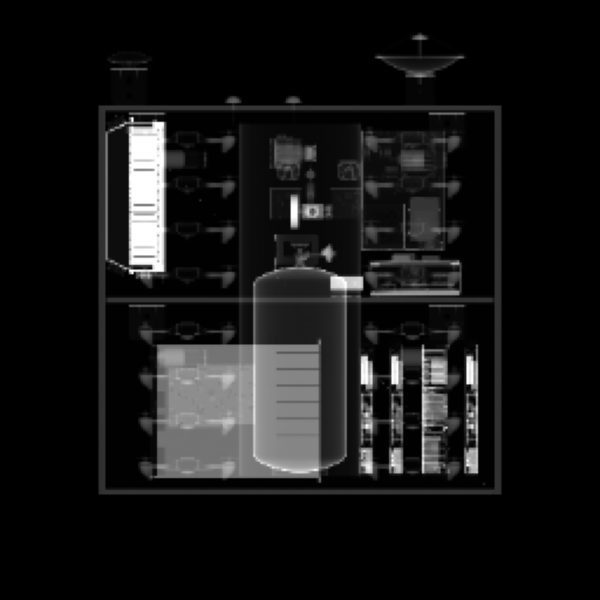}
};
\zoombox[magnification=3.5,color code=yellow,zoombox paths/.append style={line width=1pt}]{0.6,0.36}
\gdef\rowcount{2}
\zoombox[magnification=3,color code=orange,zoombox paths/.append style={line width=1pt}]{0.3,0.28}
\end{tikzpicture}
\centering \scriptsize (a) GT ($\infty$/1)
\end{subfigure} 
\begin{subfigure}[t]{0.19\textwidth}
\begin{tikzpicture}[zoomboxarray,zoomboxes yshift=1.21,zoomboxes xshift=0.,spymargin=2pt]
\pgfkeys{/tikz/zoomboxarray columns=2,/tikz/zoomboxarray rows=1,/tikz/zoomboxarray inner gap=1}
\node [image node] {
\includegraphics[width=\textwidth]{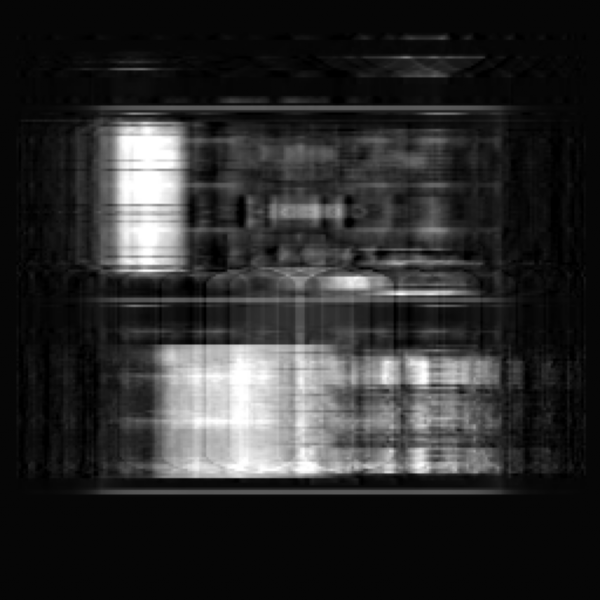}
};
\zoombox[magnification=3.5,color code=yellow,zoombox paths/.append style={line width=1pt}]{0.6,0.36}
\gdef\rowcount{2}
\zoombox[magnification=3,color code=orange,zoombox paths/.append style={line width=1pt}]{0.3,0.28}
\end{tikzpicture}
\centering \scriptsize (b) FBP (22.96/0.789)
\end{subfigure}
\begin{subfigure}[t]{0.19\textwidth}
\begin{tikzpicture}[zoomboxarray,zoomboxes yshift=1.21,zoomboxes xshift=0.,spymargin=2pt]
\pgfkeys{/tikz/zoomboxarray columns=2,/tikz/zoomboxarray rows=1,/tikz/zoomboxarray inner gap=1}
\node [image node] {
\includegraphics[width=\textwidth]{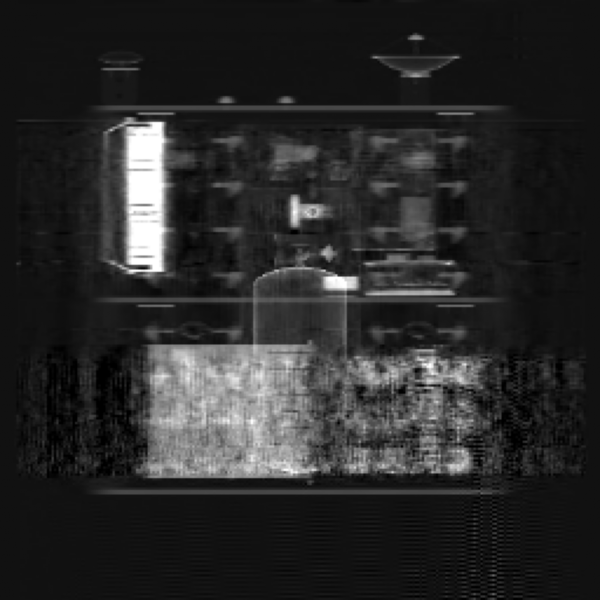}
};
\zoombox[magnification=3.5,color code=yellow,zoombox paths/.append style={line width=1pt}]{0.6,0.36}
\gdef\rowcount{2}
\zoombox[magnification=3,color code=orange,zoombox paths/.append style={line width=1pt}]{0.3,0.28}
\end{tikzpicture}
\centering \scriptsize (c) ASeqDIP (29.08/0.932)
\end{subfigure}
\begin{subfigure}[t]{0.19\textwidth}
\begin{tikzpicture}[zoomboxarray,zoomboxes yshift=1.21,zoomboxes xshift=0.,spymargin=2pt]
\pgfkeys{/tikz/zoomboxarray columns=2,/tikz/zoomboxarray rows=1,/tikz/zoomboxarray inner gap=1}
\node [image node] {
\includegraphics[width=\textwidth]{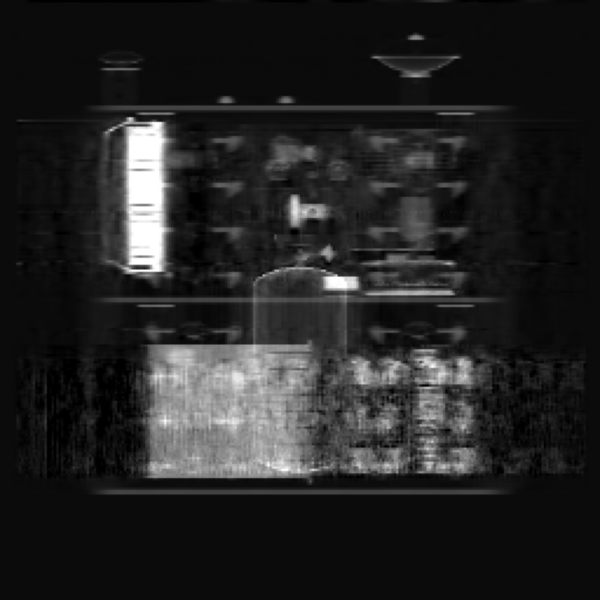}
};
\zoombox[magnification=3.5,color code=yellow,zoombox paths/.append style={line width=1pt}]{0.6,0.36}
\gdef\rowcount{2}
\zoombox[magnification=3,color code=orange,zoombox paths/.append style={line width=1pt}]{0.3,0.28}
\end{tikzpicture}
\centering \scriptsize (d) ASeqDIP + TV (30.17/0.943)
\end{subfigure}
\begin{subfigure}[t]{0.19\textwidth}
\begin{tikzpicture}[zoomboxarray,zoomboxes yshift=1.21,zoomboxes xshift=0.,spymargin=2pt]
\pgfkeys{/tikz/zoomboxarray columns=2,/tikz/zoomboxarray rows=1,/tikz/zoomboxarray inner gap=1}
\node [image node] {
\includegraphics[width=\textwidth]{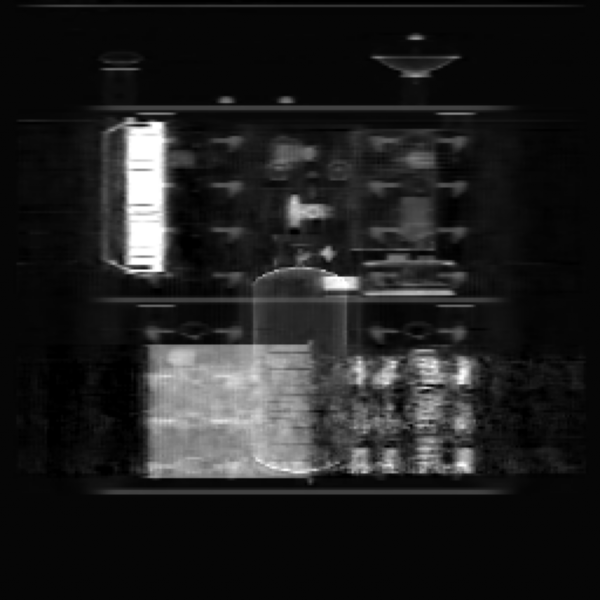}
};
\zoombox[magnification=3.5,color code=yellow,zoombox paths/.append style={line width=1pt}]{0.6,0.36}
\gdef\rowcount{2}
\zoombox[magnification=3,color code=orange,zoombox paths/.append style={line width=1pt}]{0.3,0.28}
\end{tikzpicture}
\centering \scriptsize (e) FAST-DIP (31.30/0.959)
\end{subfigure}

\vspace{-1.1in}
\begin{subfigure}[t]{0.19\textwidth}
\begin{tikzpicture}[zoomboxarray,zoomboxes yshift=1.21,zoomboxes xshift=0.,spymargin=2pt]
\pgfkeys{/tikz/zoomboxarray columns=2,/tikz/zoomboxarray rows=1,/tikz/zoomboxarray inner gap=1}
\node [image node] {
\includegraphics[width=\textwidth]{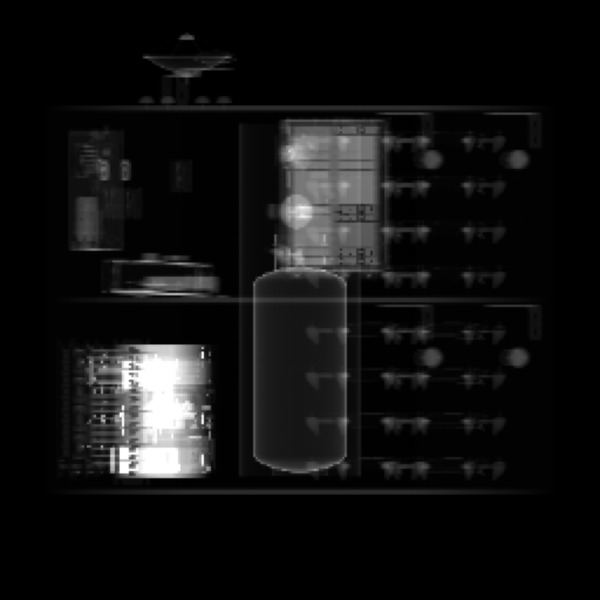}
};
\zoombox[magnification=3.5,color code=yellow,zoombox paths/.append style={line width=1pt}]{0.5,0.3}
\gdef\rowcount{2}
\zoombox[magnification=4.5,color code=orange,zoombox paths/.append style={line width=1pt}]{0.2,0.71}
\end{tikzpicture}
\centering \scriptsize (a) GT ($\infty$/1)
\end{subfigure} 
\begin{subfigure}[t]{0.19\textwidth}
\begin{tikzpicture}[zoomboxarray,zoomboxes yshift=1.21,zoomboxes xshift=0.,spymargin=2pt]
\pgfkeys{/tikz/zoomboxarray columns=2,/tikz/zoomboxarray rows=1,/tikz/zoomboxarray inner gap=1}
\node [image node] {
\includegraphics[width=\textwidth]{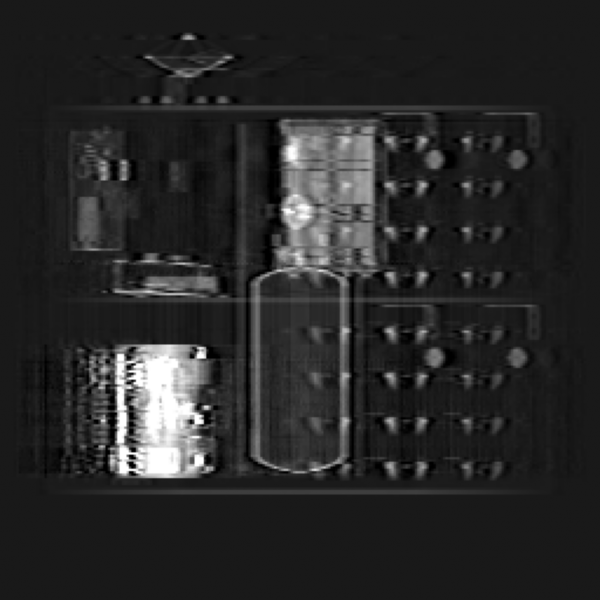}
};
\zoombox[magnification=3.5,color code=yellow,zoombox paths/.append style={line width=1pt}]{0.5,0.3}
\gdef\rowcount{2}
\zoombox[magnification=4.5,color code=orange,zoombox paths/.append style={line width=1pt}]{0.2,0.71}
\end{tikzpicture}
\centering \scriptsize (b) FBP (29.85/0.954)
\end{subfigure}
\begin{subfigure}[t]{0.19\textwidth}
\begin{tikzpicture}[zoomboxarray,zoomboxes yshift=1.21,zoomboxes xshift=0.,spymargin=2pt]
\pgfkeys{/tikz/zoomboxarray columns=2,/tikz/zoomboxarray rows=1,/tikz/zoomboxarray inner gap=1}
\node [image node] {
\includegraphics[width=\textwidth]{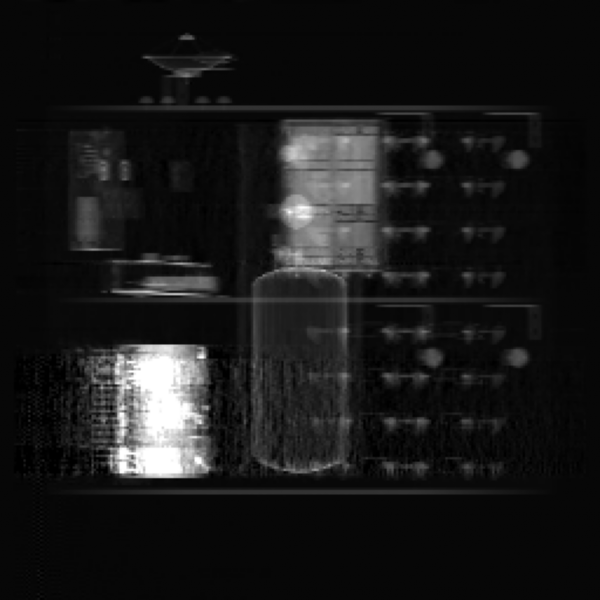}
};
\zoombox[magnification=3.5,color code=yellow,zoombox paths/.append style={line width=1pt}]{0.5,0.3}
\gdef\rowcount{2}
\zoombox[magnification=4.5,color code=orange,zoombox paths/.append style={line width=1pt}]{0.2,0.71}
\end{tikzpicture}
\centering \scriptsize (c) ASeqDIP (34.19/0.987)
\end{subfigure}
\begin{subfigure}[t]{0.19\textwidth}
\begin{tikzpicture}[zoomboxarray,zoomboxes yshift=1.21,zoomboxes xshift=0.,spymargin=2pt]
\pgfkeys{/tikz/zoomboxarray columns=2,/tikz/zoomboxarray rows=1,/tikz/zoomboxarray inner gap=1}
\node [image node] {
\includegraphics[width=\textwidth]{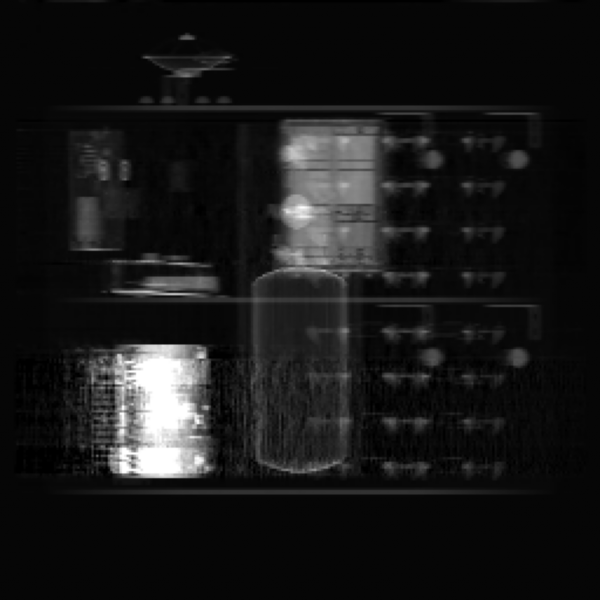}
};
\zoombox[magnification=3.5,color code=yellow,zoombox paths/.append style={line width=1pt}]{0.5,0.3}
\gdef\rowcount{2}
\zoombox[magnification=4.5,color code=orange,zoombox paths/.append style={line width=1pt}]{0.2,0.71}
\end{tikzpicture}
\centering \scriptsize (d) ASeqDIP + TV (34.52/0.988)
\end{subfigure}
\begin{subfigure}[t]{0.19\textwidth}
\begin{tikzpicture}[zoomboxarray,zoomboxes yshift=1.21,zoomboxes xshift=0.,spymargin=2pt]
\pgfkeys{/tikz/zoomboxarray columns=2,/tikz/zoomboxarray rows=1,/tikz/zoomboxarray inner gap=1}
\node [image node] {
\includegraphics[width=\textwidth]{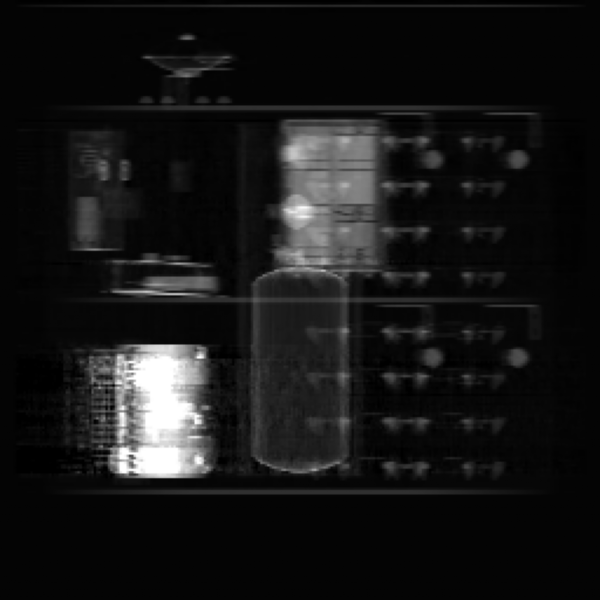}
};
\zoombox[magnification=3.5,color code=yellow,zoombox paths/.append style={line width=1pt}]{0.5,0.3}
\gdef\rowcount{2}
\zoombox[magnification=4.5,color code=orange,zoombox paths/.append style={line width=1pt}]{0.2,0.71}
\end{tikzpicture}
\centering \scriptsize (e) FAST-DIP (34.70/0.988)
\end{subfigure}
\caption{Visual comparison of novel views obtained from sparse-view 3D reconstructions of simulated XENA data \cite{xenadarpawebsite}. Zoomed regions highlight the differences in structural recovery and artifact suppression.}
    \label{fig: 3D reconstruction}
\end{figure*}
\begin{figure*}[t]
\begin{subfigure}[t]{0.19\textwidth}
\begin{tikzpicture}[zoomboxarray,zoomboxes yshift=1.21,zoomboxes xshift=0.,spymargin=2pt]
\pgfkeys{/tikz/zoomboxarray columns=2,/tikz/zoomboxarray rows=1,/tikz/zoomboxarray inner gap=1}
\node [image node] {
\includegraphics[width=\textwidth]{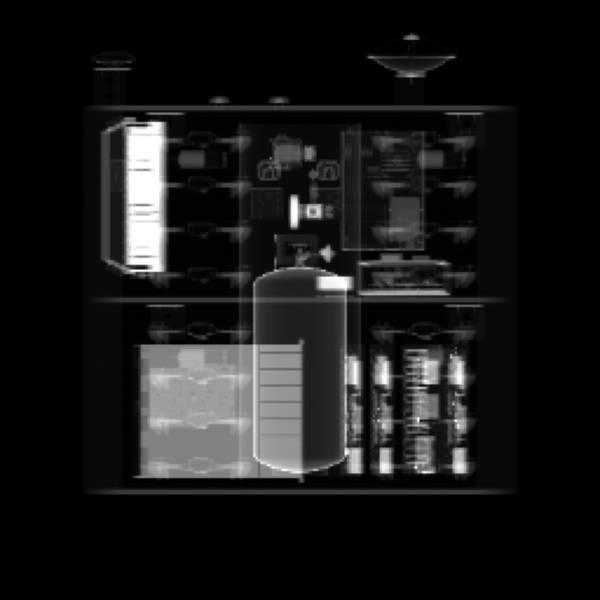}
};
\zoombox[magnification=4,color code=yellow,zoombox paths/.append style={line width=1pt}]{0.7,0.6}
\gdef\rowcount{2}
\zoombox[magnification=5.5,color code=orange,zoombox paths/.append style={line width=1pt}]{0.335,0.23}
\end{tikzpicture}
\centering \scriptsize (a) GT ($\infty$/1)
\end{subfigure} 
\begin{subfigure}[t]{0.19\textwidth}
\begin{tikzpicture}[zoomboxarray,zoomboxes yshift=1.21,zoomboxes xshift=0.,spymargin=2pt]
\pgfkeys{/tikz/zoomboxarray columns=2,/tikz/zoomboxarray rows=1,/tikz/zoomboxarray inner gap=1}
\node [image node] {
\includegraphics[width=\textwidth]{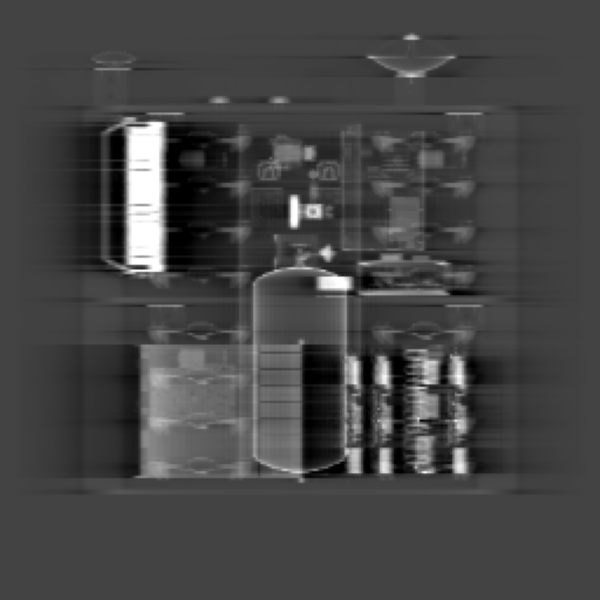}
};
\zoombox[magnification=4,color code=yellow,zoombox paths/.append style={line width=1pt}]{0.7,0.6}
\gdef\rowcount{2}
\zoombox[magnification=5.5,color code=orange,zoombox paths/.append style={line width=1pt}]{0.335,0.23}
\end{tikzpicture}
\centering \scriptsize (b) FBP (23.42/0.866)
\end{subfigure}
\begin{subfigure}[t]{0.19\textwidth}
\begin{tikzpicture}[zoomboxarray,zoomboxes yshift=1.21,zoomboxes xshift=0.,spymargin=2pt]
\pgfkeys{/tikz/zoomboxarray columns=2,/tikz/zoomboxarray rows=1,/tikz/zoomboxarray inner gap=1}
\node [image node] {
\includegraphics[width=\textwidth]{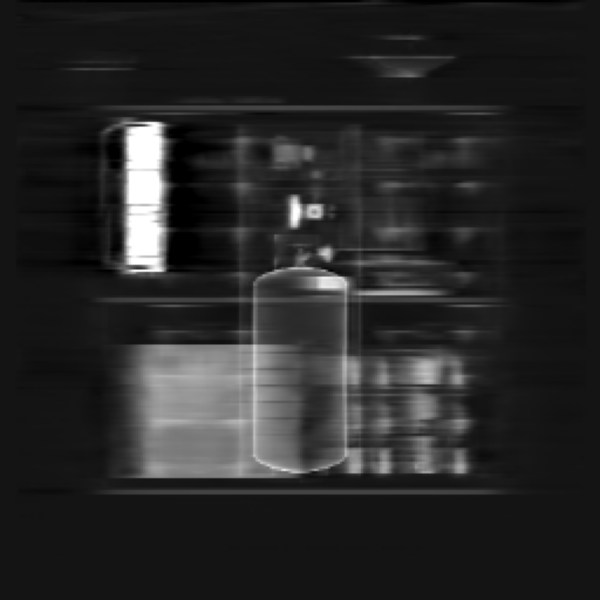}
};
\zoombox[magnification=4,color code=yellow,zoombox paths/.append style={line width=1pt}]{0.7,0.6}
\gdef\rowcount{2}
\zoombox[magnification=5.5,color code=orange,zoombox paths/.append style={line width=1pt}]{0.335,0.23}
\end{tikzpicture}
\centering\scriptsize (c) ASeqDIP (25.74/0.887)
\end{subfigure}
\begin{subfigure}[t]{0.19\textwidth}
\begin{tikzpicture}[zoomboxarray,zoomboxes yshift=1.21,zoomboxes xshift=0.,spymargin=2pt]
\pgfkeys{/tikz/zoomboxarray columns=2,/tikz/zoomboxarray rows=1,/tikz/zoomboxarray inner gap=1}
\node [image node] {
\includegraphics[width=\textwidth]{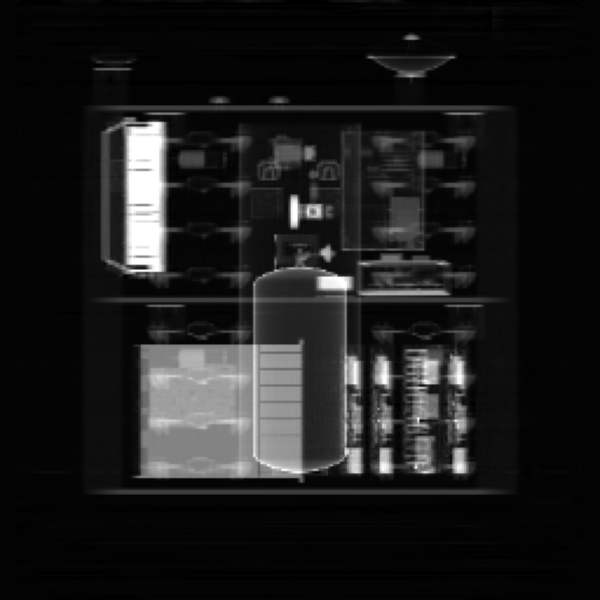}
};
\zoombox[magnification=4,color code=yellow,zoombox paths/.append style={line width=1pt}]{0.7,0.6}
\gdef\rowcount{2}
\zoombox[magnification=5.5,color code=orange,zoombox paths/.append style={line width=1pt}]{0.335,0.23}
\end{tikzpicture}
\centering \scriptsize (d) ASeqDIP + TV (42.84/0.996)
\end{subfigure}
\begin{subfigure}[t]{0.19\textwidth}
\begin{tikzpicture}[zoomboxarray,zoomboxes yshift=1.21,zoomboxes xshift=0.,spymargin=2pt]
\pgfkeys{/tikz/zoomboxarray columns=2,/tikz/zoomboxarray rows=1,/tikz/zoomboxarray inner gap=1}
\node [image node] {
\includegraphics[width=\textwidth]{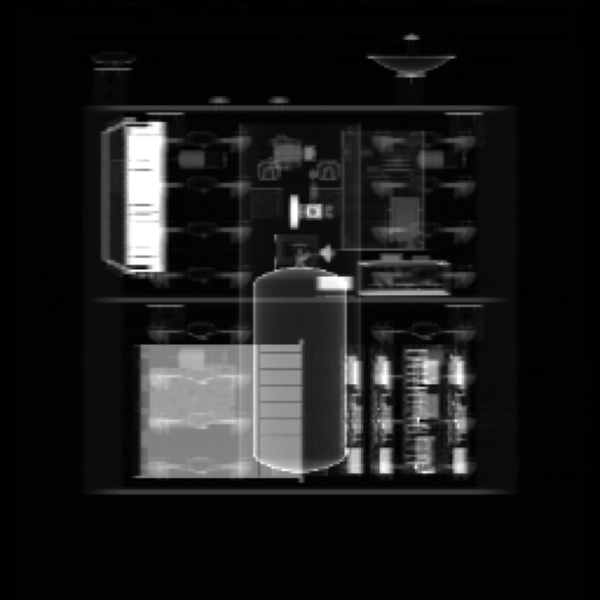}
};
\zoombox[magnification=4,color code=yellow,zoombox paths/.append style={line width=1pt}]{0.7,0.6}
\gdef\rowcount{2}
\zoombox[magnification=5.5,color code=orange,zoombox paths/.append style={line width=1pt}]{0.335,0.23}
\end{tikzpicture}
\centering \scriptsize (e) FAST-DIP (45.65/0.999)
\end{subfigure}

\vspace{-1.1in}
\begin{subfigure}[t]{0.19\textwidth}
\begin{tikzpicture}[zoomboxarray,zoomboxes yshift=1.21,zoomboxes xshift=0.,spymargin=2pt]
\pgfkeys{/tikz/zoomboxarray columns=2,/tikz/zoomboxarray rows=1,/tikz/zoomboxarray inner gap=1}
\node [image node] {
\includegraphics[width=\textwidth]{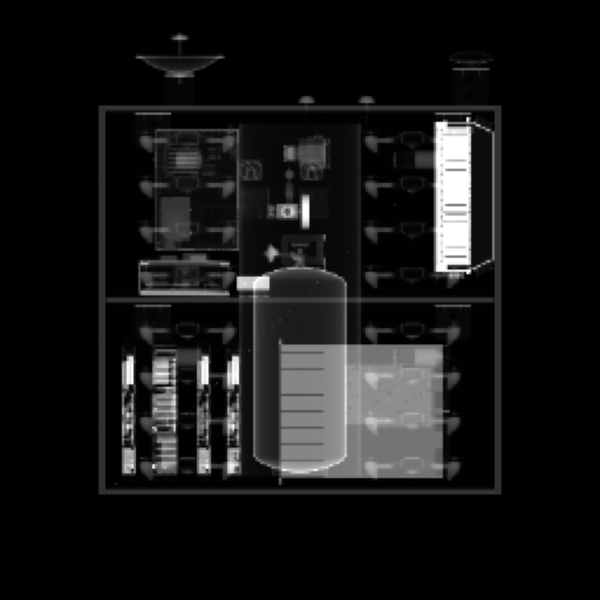}
};
\zoombox[magnification=5.5,color code=yellow,zoombox paths/.append style={line width=1pt}]{0.69,0.23}
\gdef\rowcount{2}
\zoombox[magnification=6,color code=orange,zoombox paths/.append style={line width=1pt}]{0.328,0.23}
\end{tikzpicture}
\centering \scriptsize (a) GT ($\infty$/1)
\end{subfigure} 
\begin{subfigure}[t]{0.19\textwidth}
\begin{tikzpicture}[zoomboxarray,zoomboxes yshift=1.21,zoomboxes xshift=0.,spymargin=2pt]
\pgfkeys{/tikz/zoomboxarray columns=2,/tikz/zoomboxarray rows=1,/tikz/zoomboxarray inner gap=1}
\node [image node] {
\includegraphics[width=\textwidth]{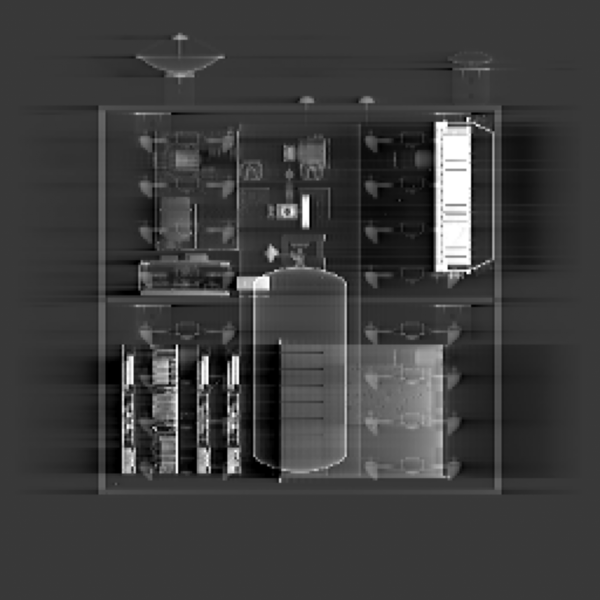}
};
\zoombox[magnification=5.5,color code=yellow,zoombox paths/.append style={line width=1pt}]{0.69,0.23}
\gdef\rowcount{2}
\zoombox[magnification=6,color code=orange,zoombox paths/.append style={line width=1pt}]{0.328,0.23}
\end{tikzpicture}
\centering \scriptsize (b) FBP (21.55/0.813)
\end{subfigure}
\begin{subfigure}[t]{0.19\textwidth}
\begin{tikzpicture}[zoomboxarray,zoomboxes yshift=1.21,zoomboxes xshift=0.,spymargin=2pt]
\pgfkeys{/tikz/zoomboxarray columns=2,/tikz/zoomboxarray rows=1,/tikz/zoomboxarray inner gap=1}
\node [image node] {
\includegraphics[width=\textwidth]{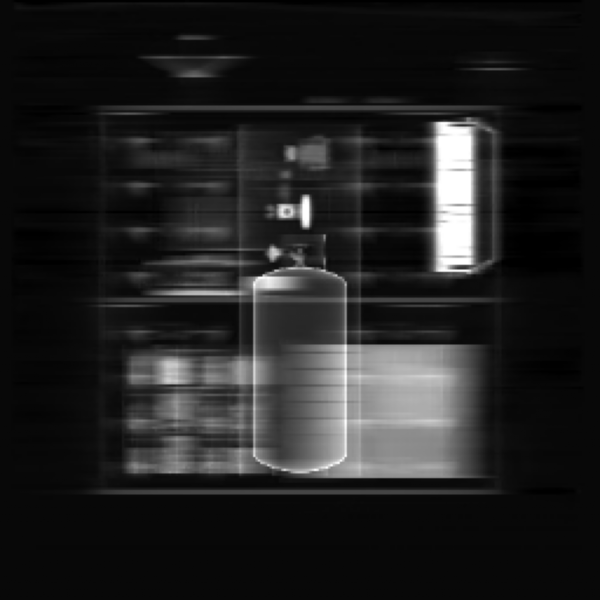}
};
\zoombox[magnification=5.5,color code=yellow,zoombox paths/.append style={line width=1pt}]{0.69,0.23}
\gdef\rowcount{2}
\zoombox[magnification=6,color code=orange,zoombox paths/.append style={line width=1pt}]{0.328,0.23}
\end{tikzpicture}
\centering \scriptsize (c) ASeqDIP (23.24/0.838)
\end{subfigure}
\begin{subfigure}[t]{0.19\textwidth}
\begin{tikzpicture}[zoomboxarray,zoomboxes yshift=1.21,zoomboxes xshift=0.,spymargin=2pt]
\pgfkeys{/tikz/zoomboxarray columns=2,/tikz/zoomboxarray rows=1,/tikz/zoomboxarray inner gap=1}
\node [image node] {
\includegraphics[width=\textwidth]{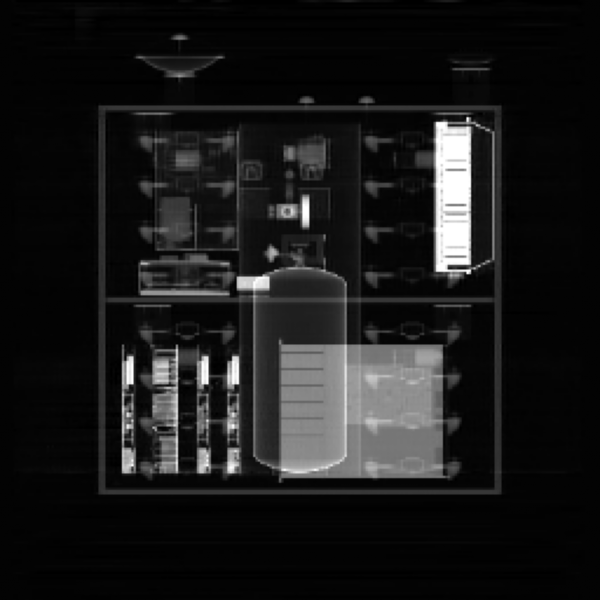}
};
\zoombox[magnification=5.5,color code=yellow,zoombox paths/.append style={line width=1pt}]{0.69,0.23}
\gdef\rowcount{2}
\zoombox[magnification=6,color code=orange,zoombox paths/.append style={line width=1pt}]{0.328,0.23}
\end{tikzpicture}
\centering \scriptsize (d) ASeqDIP + TV (34.18/0.990)
\end{subfigure}
\begin{subfigure}[t]{0.19\textwidth}
\begin{tikzpicture}[zoomboxarray,zoomboxes yshift=1.21,zoomboxes xshift=0.,spymargin=2pt]
\pgfkeys{/tikz/zoomboxarray columns=2,/tikz/zoomboxarray rows=1,/tikz/zoomboxarray inner gap=1}
\node [image node] {
\includegraphics[width=\textwidth]{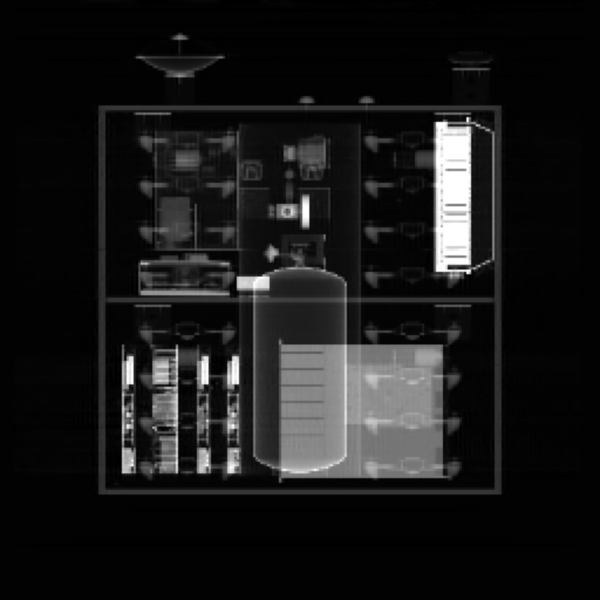}
};
\zoombox[magnification=5.5,color code=yellow,zoombox paths/.append style={line width=1pt}]{0.69,0.23}
\gdef\rowcount{2}
\zoombox[magnification=6,color code=orange,zoombox paths/.append style={line width=1pt}]{0.328,0.23}
\end{tikzpicture}
\centering \scriptsize (e) FAST-DIP (34.26/0.992)
\end{subfigure}
\caption{Visual comparison of novel views from limited-angle 3D reconstructions of simulated XENA data \cite{xenadarpawebsite}. Due to the missing angular information, this setting is more severely ill-posed than sparse-view reconstruction. The zoomed regions highlight the ability of different methods to recover structural details under strong artifacts.}
    \label{fig: 3D limit reconstruction}
\end{figure*}

\begin{table*}[t]
\centering
\caption{
Ablation study of the TV regularization parameter $\lambda$ in ASeqDIP+TV for given-view and novel-view reconstruction under sparse-view and limited-angle settings.
The best result is shown in bold, and the second-best result is underlined.
}
\label{tab:Aseq-DIP + TV ablation combined}
\small
\setlength{\tabcolsep}{4.8pt}
\renewcommand{\arraystretch}{1.15}
\begin{tabular}{lcccccccc}
\toprule
\multirow{3}{*}{\diagbox[width=7.0em,height=3.2\line]{$\lambda$}{Setting}}
& \multicolumn{4}{c}{Sparse Views}
& \multicolumn{4}{c}{Limited Angles} \\
\cmidrule(lr){2-5}
\cmidrule(lr){6-9}
& \multicolumn{2}{c}{Given Views}
& \multicolumn{2}{c}{Novel Views}
& \multicolumn{2}{c}{Given Views}
& \multicolumn{2}{c}{Novel Views} \\
\cmidrule(lr){2-3}
\cmidrule(lr){4-5}
\cmidrule(lr){6-7}
\cmidrule(lr){8-9}
& PSNR $\uparrow$ & MS-SSIM $\uparrow$
& PSNR $\uparrow$ & MS-SSIM $\uparrow$
& PSNR $\uparrow$ & MS-SSIM $\uparrow$
& PSNR $\uparrow$ & MS-SSIM $\uparrow$ \\
\midrule
$10^{-1}$
& 30.35 & 0.959
& 27.43 & 0.933
& 27.03 & 0.865
& 16.37 & 0.410 \\

$10^{-2}$
& 41.50 & \underline{0.995}
& 32.14 & \underline{0.970}
& 36.17 & 0.979
& 23.48 & 0.802 \\

$\mathbf{10^{-3}}$
& \textbf{42.81} & \textbf{0.996}
& \textbf{32.49} & \textbf{0.973}
& \textbf{42.73} & \textbf{0.996}
& \textbf{27.30} & \textbf{0.865} \\

$8\times10^{-4}$
& \underline{42.47} & \textbf{0.996}
& \underline{32.21} & \underline{0.970}
& \underline{38.79} & \underline{0.992}
& \underline{24.35} & \underline{0.822} \\
\bottomrule
\end{tabular}
\end{table*}

\begin{table*}[t]
\centering
\caption{
Ablation study of the parameter $\gamma$ in the proposed FAST-DIP for given-view and novel-view reconstruction under sparse-view and limited-angle settings.
The best result is shown in bold, and the second-best result is underlined.
}
\label{tab:Ours ablation combined}
\small
\setlength{\tabcolsep}{4.8pt}
\renewcommand{\arraystretch}{1.15}
\begin{tabular}{lcccccccc}
\toprule
\multirow{3}{*}{\diagbox[width=7.0em,height=3.2\line]{$\gamma$}{Setting}}
& \multicolumn{4}{c}{Sparse Views}
& \multicolumn{4}{c}{Limited Angles} \\
\cmidrule(lr){2-5}
\cmidrule(lr){6-9}
& \multicolumn{2}{c}{Given Views}
& \multicolumn{2}{c}{Novel Views}
& \multicolumn{2}{c}{Given Views}
& \multicolumn{2}{c}{Novel Views} \\
\cmidrule(lr){2-3}
\cmidrule(lr){4-5}
\cmidrule(lr){6-7}
\cmidrule(lr){8-9}
& PSNR $\uparrow$ & MS-SSIM $\uparrow$
& PSNR $\uparrow$ & MS-SSIM $\uparrow$
& PSNR $\uparrow$ & MS-SSIM $\uparrow$
& PSNR $\uparrow$ & MS-SSIM $\uparrow$ \\
\midrule
$3\times10^{-1}$
& 30.03 & 0.909
& 29.67 & 0.902
& 29.22 & 0.911
& 26.07 & 0.818 \\

$10^{-1}$
& 38.29 & 0.989
& \underline{32.12} & \underline{0.971}
& \underline{38.61} & \underline{0.987}
& \underline{27.14} & \underline{0.860} \\

$\mathbf{10^{-2}}$
& \textbf{43.33} & \textbf{0.996}
& \textbf{32.92} & \textbf{0.976}
& \textbf{45.11} & 0.979
& \textbf{27.46} & \textbf{0.865} \\

$8\times10^{-3}$
& \underline{41.61} & \underline{0.995}
& 31.40 & \underline{0.971}
& 37.75 & \textbf{0.988}
& 26.68 & 0.852 \\
\bottomrule
\end{tabular}
\end{table*}

\subsection{3D Reconstruction}
To evaluate the proposed FAST-DIP, our second experiment is 3D CT reconstruction. 
More specifically, we experiment with sparse-views and limited-angle 3D CT reconstruction under a parallel beam CT geometry, utilizing simulated projection measurements directly provided by the XENA dataset \cite{xenadarpawebsite}. 
We evaluate reprojections of 3D reconstructions along given views and novel views to check the quality achieved by different methods.

For the sparse-view setting, we optimize the network using 20 projection views evenly distributed over a 180 degree range. For the limited-angle setting, we use 20 projection views distributed within a restricted angular range of 20 degrees, simulating realistic constrained scanning scenarios.
Before optimization, the projection data is initially reconstructed into a 3D volume using FBP. The reconstruction process is divided into two stages. First, the network parameters $\phi$ are updated sliceby slice using the Adam optimizer with a learning rate of 1e-4, thereby constructing an initial 3D reconstructed volume. In the second stage, this volume is used as the initialization for regularization, and the variable $z$ is optimized via gradient descent by introducing a fractional penalty term. We set the penalty weight to $\gamma = 0.01$ and the step size to $\beta = 0.1$.
We evaluate the FBP, ASeqDIP \cite{alkhouri2024image}, ASeqDIP with smoothed TV regularization (ASeqDIP+TV), and FAST-DIP for 3D reconstruction. For ASeqDIP+TV, we set $M^{k-1}=1$ in our Algorithm \ref{algo:mm_agd}, other settings are the same as the proposed FAST-DIP. The numerical results are presented in Table \ref{tab:3Dreconstruction-results}. Overall, our FAST-DIP achieves the best performance across most evaluation settings. In the sparse-view setting, FAST-DIP consistently achieves the best or second-best performance for both given-view and novel-view evaluations. For the more challenging limited-angle setting, the proposed FAST-DIP demonstrates clear advantages. It achieves the highest PSNR in both the given-view and novel-view cases and the best MS-SSIM in most settings, suggesting stronger robustness under severely incomplete measurements.

We present visual comparisons of novel sparse-view reconstructions in Fig. \ref{fig: 3D reconstruction}. ASeqDIP reduces part of these artifacts but still produces noisy textures and blurred structures.
For the ASeqDIP + TV method, we modify $M^{k-1}=1$ in the proposed
Algorithm \ref{algo:mm_agd}. Compared with the ASeqDIP results, the ASeqDIP + TV further suppresses noise. 
The proposed FAST-DIP produces more stable reconstructions with improved structural consistency. As highlighted in the zoomed regions (orange and yellow boxes), our approach better preserves small structures and reduces streaking artifacts compared to competing methods.

Besides, the novel limited-angle results are reported in Fig. \ref{fig: 3D limit reconstruction}. Compared with sparse-view reconstruction, the limited-angle setting is more severely ill-posed due to the missing angular information. Thus, all competing methods exhibit stronger artifacts and structural distortions. Nevertheless, the proposed FAST-DIP still maintains clear structural boundaries and more faithful geometric patterns, demonstrating stronger robustness under extremely incomplete measurements. Moreover, the improved reconstruction quality on novel views indicates that FAST-DIP generalizes better to unseen viewpoints, suggesting that it captures more intrinsic structural information of the underlying 3D object.

\subsection{Model Analysis}

To analyze the effect of the regularization parameters, we conducted ablation studies on Algorithm \ref{algo:mm_agd} with $M^{k-1}=1$ and the original Algorithm \ref{algo:mm_agd}, which is the ASeqDIP + TV and FAST-DIP. 

For ASeqDIP + TV, we evaluate different values of the regularization parameter $\lambda$. As shown in Table \ref{tab:Aseq-DIP + TV ablation combined}, moderate regularization leads to the best reconstruction quality. In particular, $\lambda = 0.001$ consistently achieves the highest PSNR and MS-SSIM across both sparse-view and limited-angle settings for given and novel views. When $\lambda$ is too large, the reconstruction suffers from excessive smoothing, which degrades structural fidelity. Conversely, overly small values weaken the regularization effect and lead to performance degradation.

For the proposed FAST-DIP, we further study the influence of the parameter $\gamma$. The results in Table \ref{tab:Ours ablation combined} show that $\gamma = 0.01$ yields the best overall performance in both sparse-view and limited-angle scenarios. Larger values tend to over-regularize the reconstruction, while smaller values reduce the effectiveness of the structural constraint. 
\section{Conclusion}
\label{sec:conclusion}
In this paper, we proposed a novel reconstruction framework that integrates deep image priors with fractional-gradient sparsity for 3D volumetric reconstruction. Building upon the input-adaptive modeling strategy of autoencoding sequential deep image prior, the proposed method leverages slice-wise deep image priors while explicitly enforcing structural consistency across slices. To better capture sparse inter-slice variations, we introduced an $\ell_1/\ell_2$-based fractional sparsity regularization on the slice-direction gradients. The proposed prior is scale-invariant and helps mitigate the over-smoothing and staircasing artifacts.
We further provided theoretical analysis for the proposed optimization scheme under the majorization–minimization framework, including monotonic descent and subsequence convergence guarantees under the Kurdyka–Łojasiewicz property. 
Experimental results on both 2D denoising and 3D reconstruction tasks demonstrate that the proposed method achieves improved reconstruction quality compared with existing DIP-based approaches while effectively preserving structural details across slices.

\section*{Acknowledgments}
This research was developed with funding from the Defense Advanced Research Projects Agency (DARPA).
The views, opinions and/or findings expressed are those of the authors and should not be interpreted as representing the official views or policies of DARPA or the U.S. Government.

\bibliography{main}
\newpage
\appendix
\onecolumn
\par\noindent\rule{\textwidth}{1pt}
\begin{center}
{\Large \bf Appendix}
\end{center}
\vspace{-0.1in}
\par\noindent\rule{\textwidth}{1pt}
\appendix

\section{Appendix: Theoretical Proofs}
\label{sec:appendix}

In this section, we provide detailed proofs for the theoretical claims in the main manuscript.\blfootnote{DARPA Distribution Statement A. Approved for public release: distribution is
unlimited.}

\subsection{Proof of Lemma \ref{lemma1} (Majorization of the Smoothed Numerator)}

\begin{proof}
We consider the smoothed numerator
\begin{equation}
N_{\delta}(\mathbf{z}) = \sum_i \sqrt{|(\nabla \mathbf{z})_i|^2 + \delta},
\end{equation}
where $\delta > 0$ is a smoothing parameter.

For each spatial index $i$, define the scalar function
\begin{equation}
h(t) = \sqrt{t+\delta}, \qquad t \ge 0.
\end{equation}
Its first and second derivatives are
\begin{equation}
h'(t) = \frac{1}{2\sqrt{t+\delta}},
\qquad
h''(t) = -\frac{1}{4}(t+\delta)^{-3/2} < 0.
\end{equation}
Hence, $h$ is strictly concave on $[0,\infty)$.

By the supporting-hyperplane property of concave functions, for any $t,t_0 \ge 0$,
\begin{equation}
h(t) \le h(t_0) + h'(t_0)(t-t_0).
\end{equation}
Substituting the expression of $h$, we obtain
\begin{equation}
\sqrt{t+\delta}
\le
\sqrt{t_0+\delta}
+
\frac{1}{2\sqrt{t_0+\delta}}(t-t_0).
\end{equation}

Now let
\begin{equation}
t = |(\nabla \mathbf{z})_i|^2,
\qquad
t_0 = |(\nabla \mathbf{z}^{k-1})_i|^2.
\end{equation}
Then, for each $i$,
\begin{equation}
\sqrt{|(\nabla \mathbf{z})_i|^2+\delta}
\le
\sqrt{|(\nabla \mathbf{z}^{k-1})_i|^2+\delta}
+
\frac{|(\nabla \mathbf{z})_i|^2 - |(\nabla \mathbf{z}^{k-1})_i|^2}
{2\sqrt{|(\nabla \mathbf{z}^{k-1})_i|^2+\delta}}.
\end{equation}

Summing the above inequality over all indices $i$ gives
\begin{equation}
\begin{aligned}
N_{\delta}(\mathbf{z})
&=
\sum_i \sqrt{|(\nabla \mathbf{z})_i|^2+\delta} \\
&\le
\sum_i
\left[
\sqrt{|(\nabla \mathbf{z}^{k-1})_i|^2+\delta}
+
\frac{|(\nabla \mathbf{z})_i|^2 - |(\nabla \mathbf{z}^{k-1})_i|^2}
{2\sqrt{|(\nabla \mathbf{z}^{k-1})_i|^2+\delta}}
\right] \\
&=
\frac{1}{2}\sum_i
\frac{|(\nabla \mathbf{z})_i|^2}
{\sqrt{|(\nabla \mathbf{z}^{k-1})_i|^2+\delta}}
+ C^{k-1},
\end{aligned}
\end{equation}
where
\begin{equation}
C^{k-1}
=
\sum_i
\left(
\sqrt{|(\nabla \mathbf{z}^{k-1})_i|^2+\delta}
-
\frac{|(\nabla \mathbf{z}^{k-1})_i|^2}
{2\sqrt{|(\nabla \mathbf{z}^{k-1})_i|^2+\delta}}
\right)
\end{equation}
is independent of $\mathbf z$.
Therefore, defining
\begin{equation}
\tilde N_{\delta}(\mathbf{z}\mid \mathbf{z}^{k-1})
=
\frac{1}{2}\sum_i
\frac{|(\nabla \mathbf{z})_i|^2}
{\sqrt{|(\nabla \mathbf{z}^{k-1})_i|^2+\delta}}
+ C^{k-1},
\end{equation}
we have
\begin{equation}
N_{\delta}(\mathbf{z}) \le \tilde N_{\delta}(\mathbf{z}\mid \mathbf{z}^{k-1}),
\end{equation}
which proves the upper-bound property.

Next, we verify the tangency condition. Setting $\mathbf{z}=\mathbf{z}^{k-1}$, we have
\begin{equation}
\begin{aligned}
\tilde N_{\delta}(\mathbf{z}^{k-1}\mid \mathbf{z}^{k-1})
&=
\frac{1}{2}\sum_i
\frac{|(\nabla \mathbf{z}^{k-1})_i|^2}
{\sqrt{|(\nabla \mathbf{z}^{k-1})_i|^2+\delta}}
+ C^{k-1} \\
&=
\sum_i \sqrt{|(\nabla \mathbf{z}^{k-1})_i|^2+\delta}
=
N_{\delta}(\mathbf{z}^{k-1}).
\end{aligned}
\end{equation}
Hence the tangency condition also holds.

Finally, by defining the diagonal weight matrix
\begin{equation}
\mathbf W^{k-1}_{i,i}
=
\frac{1}{\sqrt{|(\nabla \mathbf{z}^{k-1})_i|^2+\delta}},
\end{equation}
the quadratic part of the surrogate can be written compactly as
\begin{equation}
\frac{1}{2}\sum_i
\frac{|(\nabla \mathbf{z})_i|^2}
{\sqrt{|(\nabla \mathbf{z}^{k-1})_i|^2+\delta}}
=
\frac{1}{2}(\nabla \mathbf{z})^\top \mathbf W^{k-1} (\nabla \mathbf{z}).
\end{equation}
This completes the proof.
\end{proof}

\subsection{Proof of Lemma \ref{lemma2} (Lipschitz Continuous Gradient)}

\begin{proof}
For fixed $\phi^k$ and $\mathbf z^{k-1}$, the surrogate objective is
\begin{equation}
\mathcal L_s(\phi^k,\mathbf z \mid \mathbf z^{k-1})
=
F(\phi^k,\mathbf z)
+
S(\mathbf z \mid \mathbf z^{k-1}),
\end{equation}
where
\begin{equation}
S(\mathbf z \mid \mathbf z^{k-1})
=
\frac{\gamma}{M^{k-1}} \tilde N_{\delta}(\mathbf z \mid \mathbf z^{k-1}),
~~
M^{k-1} = \|\nabla \mathbf z^{k-1}\|_2 + \epsilon.
\end{equation}
We show that each term has Lipschitz continuous gradient with respect to $\mathbf z$ on the bounded iterate set.

First, since $S(\mathbf z \mid \mathbf z^{k-1})$ is quadratic in $\mathbf z$, its Hessian is constant and given by
\begin{equation}
\nabla_{\mathbf z}^2 S(\mathbf z \mid \mathbf z^{k-1})
=
\frac{\gamma}{M^{k-1}}
\nabla^\top \mathbf W^{k-1} \nabla.
\end{equation}
Because $\epsilon>0$, we have $M^{k-1}\ge \epsilon$. Moreover, since $\delta>0$,
\begin{equation}
\mathbf W^{k-1}_{i,i}
=
\frac{1}{\sqrt{|(\nabla \mathbf z^{k-1})_i|^2+\delta}}
\le \frac{1}{\sqrt{\delta}}.
\end{equation}
Hence
\begin{equation}
\left\|
\nabla_{\mathbf z}^2 S(\mathbf z \mid \mathbf z^{k-1})
\right\|_2
\le
\frac{\gamma}{\epsilon \sqrt{\delta}}
\|\nabla^\top \nabla\|_2,
\end{equation}
which is finite because the discrete gradient operator is linear and bounded in finite dimensions.

Next, consider
\begin{equation}
F(\phi^k,\mathbf z)
=
\|\mathbf A f_{\phi^k}(\mathbf z)-\mathbf y\|_2^2
+
\lambda \|\mathbf z - f_{\phi^k}(\mathbf z)\|_2^2.
\end{equation}
By assumption, for fixed $\phi^k$, the mapping $f_{\phi^k}(\mathbf z)$ is continuously differentiable and has bounded Jacobian and Hessian on the bounded iterate set. Since $\mathbf A$ is linear and bounded, standard chain-rule calculations imply that the Hessian $\nabla_{\mathbf z}^2 F(\phi^k,\mathbf z)$ is bounded on the same set. Therefore, there exists a constant $L_F>0$ such that
\begin{equation}
\|\nabla_{\mathbf z}^2 F(\phi^k,\mathbf z)\|_2 \le L_F.
\end{equation}
Combining the two bounds, we obtain
\begin{equation}
\|\nabla_{\mathbf z}^2 \mathcal L_s(\phi^k,\mathbf z \mid \mathbf z^{k-1})\|_2
\le
L_F + \frac{\gamma}{\epsilon \sqrt{\delta}}\|\nabla^\top \nabla\|_2
=: L.
\end{equation}
Hence the gradient $\nabla_{\mathbf z} \mathcal L_s(\phi^k,\mathbf z \mid \mathbf z^{k-1})$ is $L$-Lipschitz continuous on the bounded iterate set.
\end{proof}

\subsection{Proof of Proposition \ref{proposition1} (Monotonic Descent of the Surrogate Objective)}

\begin{proof}
By Lemma \ref{lemma2}, the gradient of the surrogate objective
\[
\mathbf z \mapsto \mathcal L_s(\phi^k,\mathbf z \mid \mathbf z^{k-1})
\]
is $L$-Lipschitz continuous. Therefore, by the standard descent lemma, for any $\mathbf u,\mathbf v$,
\begin{equation}
\begin{aligned}
\mathcal L_s(\phi^k,\mathbf u \mid \mathbf z^{k-1})
\le&
\mathcal L_s(\phi^k,\mathbf v \mid \mathbf z^{k-1})
+\\&
\left\langle
\nabla_{\mathbf z}\mathcal L_s(\phi^k,\mathbf v \mid \mathbf z^{k-1}),
\mathbf u-\mathbf v
\right\rangle
+
\frac{L}{2}\|\mathbf u-\mathbf v\|_2^2.
\end{aligned}
\end{equation}
Now set $
\mathbf v = \mathbf z^{k-1},
\mathbf u = \mathbf z^k
=
\mathbf z^{k-1}
-
\beta \nabla_{\mathbf z}\mathcal L_s(\phi^k,\mathbf z \mid \mathbf z^{k-1})\big|_{\mathbf z=\mathbf z^{k-1}}.$
Then
\begin{equation}
\mathbf z^k - \mathbf z^{k-1}
=
-\beta \nabla_{\mathbf z}\mathcal L_s(\phi^k,\mathbf z^{k-1} \mid \mathbf z^{k-1}).
\end{equation}
Substituting into the descent inequality yields
\begin{equation}
\begin{aligned}
&\mathcal L_s(\phi^k,\mathbf z^k \mid \mathbf z^{k-1})
\\\le&
\mathcal L_s(\phi^k,\mathbf z^{k-1} \mid \mathbf z^{k-1})
-
\beta
\|\nabla_{\mathbf z}\mathcal L_s(\phi^k,\mathbf z^{k-1} \mid \mathbf z^{k-1})\|_2^2 \\
&
+
\frac{L\beta^2}{2}
\|\nabla_{\mathbf z}\mathcal L_s(\phi^k,\mathbf z^{k-1} \mid \mathbf z^{k-1})\|_2^2 \\
=&
\mathcal L_s(\phi^k,\mathbf z^{k-1} \mid \mathbf z^{k-1})
-
\left(\beta-\frac{L\beta^2}{2}\right)
\|\nabla_{\mathbf z}\mathcal L_s(\phi^k,\mathbf z^{k-1} \mid \mathbf z^{k-1})\|_2^2.
\end{aligned}
\end{equation}
If $\beta \le 1/L$, then
\begin{equation}
\beta-\frac{L\beta^2}{2} \ge \frac{\beta}{2} >0,
\end{equation}
and therefore
\begin{equation}
\mathcal L_s(\phi^k,\mathbf z^k \mid \mathbf z^{k-1})
\le
\mathcal L_s(\phi^k,\mathbf z^{k-1} \mid \mathbf z^{k-1}).
\end{equation}
This proves the monotonic descent of the surrogate objective.
\end{proof}

\subsection{Proof Sketch of Theorem \ref{theorem1} (Subsequence Convergence under KL Framework)}

\begin{proof}
The proof follows a standard KL-based argument for surrogate-based alternating optimization.

First, by Proposition \ref{proposition1}, the $\mathbf z$-update yields a sufficient decrease of the surrogate objective at each iteration. By assumption, the $\phi$-update together with the $\mathbf z$-update defines a bounded surrogate-based alternating sequence that satisfies a sufficient decrease condition and a relative error condition.

Second, the involved functions are semi-algebraic. Indeed, the data fidelity term is polynomial in its arguments up to composition with the network mapping; the autoencoding term is quadratic; and the surrogate fractional term is a weighted quadratic function with weights determined from the previous iterate. Under the stated assumptions, the resulting surrogate-based objective belongs to the class of semi-algebraic (more generally, KL) functions.

Therefore, the abstract convergence results for KL functions apply: the sufficient decrease condition prevents oscillatory growth of the surrogate objective, the relative error condition controls the first-order residual by successive differences, and boundedness ensures the existence of accumulation points. Consequently, the sequence has finite length,
\begin{equation}
\sum_{k=1}^\infty
\|(\phi^{k+1},\mathbf z^{k+1})-(\phi^k,\mathbf z^k)\|_2
<\infty,
\end{equation}
and every accumulation point is a critical point of the surrogate-based alternating scheme.

We emphasize that this theorem concerns the surrogate-based alternating scheme induced by the IRLS-type approximation, rather than a strict global majorization of the original full ratio objective.
\end{proof}

\end{document}